\documentclass{article}
\usepackage{graphicx}
\usepackage{subfigure}
\usepackage{amsfonts}
\usepackage{amsmath}
\usepackage{amssymb}
\usepackage{url}
\usepackage{fancyhdr}
\usepackage{indentfirst}
\usepackage{enumerate}
\usepackage{cancel}
\usepackage{dsfont}
  \usepackage[colorlinks=true,citecolor=blue,hyperfootnotes=false]{hyperref}
\usepackage{amsthm}
\usepackage{color}
\usepackage{natbib}
\usepackage{comment}
\usepackage{capt-of}
\usepackage{multirow}
\def\d{\mathrm{d}}

\newcommand{\Q}{\mathcal {Q}}
\newcommand{\rbf}{\mathbf{r}}
\newcommand{\q}{\mathbf{q}}

\newcommand{\X}{\mathcal {X}}

\newcommand{\var}{\mathrm{Var}}

\newcommand{\VaR}{\mathrm{VaR}}

\newcommand{\ES}{\mathrm{ES}}

\newcommand{\E}{\mathbb{E}}

\newcommand{\R}{\mathbb{R}}

\newcommand{\p}{\mathbb{P}}

\newcommand{\Var}{\mathrm{Var}}
\newcommand{\SD}{\mathrm{SD}}

\renewcommand{\ge}{\geqslant}

\renewcommand{\geq}{\geqslant}
\renewcommand{\leq}{\leqslant}
\renewcommand{\epsilon}{\varepsilon}

\renewcommand{\cdots}{\dots}

\theoremstyle{plain}
\newtheorem{theorem}{Theorem}
\newtheorem{corollary}{Corollary}

\newtheorem{proposition}{Proposition}
\theoremstyle{definition}

\newtheorem{example}{Example}

\newtheorem{conjecture}{Conjecture}
\theoremstyle{remark}
\newtheorem{remark}{Remark}

\definecolor{editgreen}{RGB}{0,128,0}

\newcommand{\cet}{\begin{center}}
\newcommand{\ecet}{\end{center}}

\newcommand{\cR}{\mathcal{R}}

\usepackage{setspace}
\newcommand{\com}[1]{\marginpar{{\begin{minipage}{0.18\textwidth}{\setstretch{1.1} \begin{flushleft} \footnotesize \color{red}{#1} \end{flushleft} }\end{minipage}}}}

\begin{document}

\title{
Portfolio Diversification and Concentration under Dependence Uncertainty: A Majorization Approach
}
\author{
    Peng Liu\thanks{\scriptsize School of Mathematics, Statistics and Actuarial Science, University of Essex, 
    UK. Email: \texttt{peng.liu@essex.ac.uk}}
    \and
    Yang Liu\thanks{\scriptsize School of Science and Engineering, The Chinese University of Hong Kong (Shenzhen), 
    China. Email: \texttt{yangliu16@cuhk.edu.cn}}
   }

\date{}
\maketitle



\begin{abstract}
Modern portfolio theory identifies diversification as the primary tool for risk reduction. However, under model uncertainty, this cornerstone may no longer remain optimal. This paper investigates the tension between portfolio diversification and concentration under dependence uncertainty. In the absence of model uncertainty, we employ the framework of the majorization order and doubly stochastic matrices to formalize the degree of diversification, and prove that quasi-convexity is a necessary and sufficient property for a risk functional to be weakly consistent with the majorization order. We further derive worst-case risk measure inequalities and solve robust portfolio selection problems for a broad class of risk measures, including VaR, ES, Range-VaR (RVaR), and standard deviation (SD). Our results reveal a ``concentration paradox'' for many widely-used risk functionals: when the dependence structure is fully ambiguous, robust optimization often recommends concentrating investment in a single asset to hedge against the worst-case dependence scenario. {As an application, we propose a weighted robustness formulation that interpolates between a reference dependence structure and the worst-case structure. The formulation is structurally analogous to the constrained/unconstrained Expected Shortfall blend in the Fundamental Review of the Trading Book (FRTB) and provides a theoretical foundation for balancing diversification against robustness in the presence of model uncertainty.}
\end{abstract}

\textbf{Keywords:} Dependence uncertainty; Robust portfolio selection; Majorization order; Value-at-Risk (VaR); Expected Shortfall (ES); Non-convex risk functionals. 


\section{Introduction}

Diversification has been a cornerstone of modern portfolio theory since \cite{M52}, yet practitioners routinely encounter episodes where concentration (sometimes intentional, sometimes forced) plays a decisive role in outcomes. It is well documented in the literature that portfolio selection fundamentally involves a strategic trade-off between \emph{diversification} and \emph{concentration} (e.g., \cite{BGUW12}). 


{The portfolio optimization problem is formulated as follows.} Let $\mathbf{X} = (X_1, \dots, X_n)^\top$
{be the vector of negative returns of $n$ assets in the market}, and $\boldsymbol{\lambda} = (\lambda_1, \dots, \lambda_n)^\top \in \Delta_n$ denote the portfolio weights, where the decision set $\Delta_n = \{\boldsymbol{\lambda} \in [0,1]^n: \sum_{i=1}^n \lambda_i = 1\}$ is the standard simplex, implying that short-selling is not allowed. The portfolio selection problem is to maximize the risk-adjusted expected return:
\begin{align}\label{RAR}
\E \left[-\boldsymbol\lambda^\top \mathbf X \right] - \kappa \rho (\boldsymbol\lambda^\top \mathbf X),
\end{align}
where $\kappa > 0$ and $\rho: \mathcal{X} \to \mathbb{R}$ is a risk functional. {If the asset returns have the same expected value, or if the portfolio is constrained to have a fixed expected return,} then the portfolio selection problem \eqref{RAR} reduces to the following form:
\begin{equation}\label{eq:opt1}
	\inf_{\boldsymbol{\lambda} \in \Delta_n} \rho(\boldsymbol{\lambda}^\top \mathbf X).
\end{equation}
If $\rho$ is the variance, Problem \eqref{eq:opt1} becomes a classical variance minimization problem. {For example, suppose $X_i \sim \mathrm{N}(\mu_i, \sigma_i^2)$ for $i=1,2$, let $\tau$ be the correlation coefficient between $X_1$ and $X_2$, and set $\lambda_2=1-\lambda_1$. Then  
\begin{equation*}
\lambda_1^*=\frac{\sigma_2^2-\tau\sigma_1\sigma_2}{\sigma_1^2+\sigma_2^2-2\tau\sigma_1\sigma_2}
\end{equation*}
is the minimizer and gives a diversified portfolio, given a suitable interior-point condition. 

Solving Problem \eqref{eq:opt1} requires knowledge of the distribution of $\mathbf X$, yet joint distribution models are notoriously difficult to estimate accurately in practice. A central challenge arises from uncertainty in the dependence structure among assets. In many realistic settings, the marginal distributions of individual risk factors are relatively well understood, whereas their joint distribution or copula remains highly uncertain. As a consequence, misspecification of dependence can fundamentally distort portfolio risk assessment and even reverse the classical benefits of diversification.

The collapse of Silicon Valley Bank provides a striking illustration of this phenomenon. Multiple risk drivers that appeared manageable in isolation became highly aligned under stressed market conditions, amplifying losses simultaneously on both the asset and liability sides of the balance sheet. Events of this kind highlight the importance of robust approaches to risk aggregation and portfolio selection that explicitly incorporate dependence uncertainty; {see \cite{EPR13}, \cite{EWW15}, \cite{MFE15} and \cite{BLLW20} on the need to account for dependence uncertainty.}

{Motivated by these considerations, we incorporate into portfolio selection an uncertainty set with fixed marginal distributions and a completely unknown dependence structure:}
\begin{equation}\label{eq:opt2}
	\inf_{\boldsymbol{\lambda} \in \Delta_n} \sup_{F_{\mathbf X} \in \mathcal{E}_n(\mathbf F)} \rho(\boldsymbol{\lambda}^\top \mathbf X),
\end{equation}
where  $\mathcal{E}_n(\mathbf F) = \{F_{\mathbf X}: X_i \sim F_i, \, i = 1, \dots, n\}$.
 If $\rho$ is the variance, a coherent or convex risk measure, or a negative expected concave utility functional, the inner problem $\sup_{F_{\mathbf X} \in \mathcal{E}_n(\mathbf F)} \rho(\boldsymbol{\lambda}^\top \mathbf X)$ can be analytically solved, and the worst-case dependence structure of $\mathbf X$ is comonotonic. That is, the optimization problem can be rewritten as
\begin{equation}\label{eq:opt3}
	\inf_{\boldsymbol{\lambda} \in \Delta_n} \rho(\boldsymbol{\lambda}^\top \mathbf Z), \quad \text{where } \mathbf Z = (F_1^{-1}(U), \dots, F_n^{-1}(U)) \text{ for some } U \sim \mathrm{U}[0, 1],
\end{equation}
where $F_i^{-1}$ is the quantile function of $F_i$.
In the case where all marginal distributions are from a location-scale family, the minimizer usually corresponds to a \emph{concentrated portfolio}. Indeed, in this case, the correlation between any two components in $\mathbf Z$ is 1. Assuming $\rho$ is the variance, we minimize $\Var(\boldsymbol{\lambda}^\top \mathbf Z) = \boldsymbol{\lambda}^\top (\mathbf a \mathbf a^\top) \boldsymbol{\lambda}$, where $\mathbf a$ is the vector of standard deviations of $\mathbf Z$. Consequently, a minimizer to \eqref{eq:opt3} is to invest solely in the asset with the smallest standard deviation.

In this paper, we explore results for Problems \eqref{eq:opt1} and \eqref{eq:opt2} under different risk functionals $\rho$ and dependence uncertainty. Our main focus is the impact of \emph{dependence uncertainty} on portfolio selection. Our results extend those of \cite{CLLW22}, who used a majorization-order-based approach to discuss diversification under dependence uncertainty with the same marginal distribution. We advance this direction by shifting the focus from asset distributions to the geometry of weights, using doubly stochastic matrices as the mathematical vehicle for averaging portfolio positions. {This approach allows us to investigate the conditions under which a move toward diversification (in the sense of majorization) can reduce risk.} We study these questions across a broad range of risk functionals, including Value-at-Risk ($\VaR$), Expected Shortfall ($\ES$), Range-Value-at-Risk (RVaR), standard deviation (SD), and other risk functionals satisfying certain properties, under several uncertainty regimes: a known distribution, completely unknown dependence structure, and uncertainty sets characterized by Wasserstein distance or moment information.

{In fact, the diversification effect in portfolio selection with respect to the majorization order has been studied extensively in the literature in the absence of model uncertainty. The diversification effect with respect to the majorization order was studied in \cite{IB05} for portfolios consisting of iid positive one-sided stable random variables and in \cite{IB09} for portfolios consisting of iid random variables following convolutions of symmetric stable distributions. More recently, this effect has been studied in a series of papers, including \cite{CEW25} and \cite{CHWZ25}, for portfolios consisting of iid (or negatively dependent) Pareto random variables with infinite mean or other random variables with infinite mean. The message of this literature is that diversification may increase portfolio risk when the underlying random variables have infinite means. We will show later that, in the presence of complete dependence uncertainty, the same phenomenon can arise even when the underlying random variables have finite means. 

There is a rich literature on portfolio selection using robust risk measures under model uncertainty. Portfolio selection for VaR, ES, and general distortion risk measures or distortion riskmetrics under moment uncertainty sets has been studied, e.g., in \cite{EOO03}, \cite{CHZ11}, \cite{LSWY18}, \cite{L18}, and \cite{PWW25}. Portfolio optimization for $\ES$ under Wasserstein uncertainty sets was studied in \cite{PW07} and \cite{EK18}. The mean-variance portfolio selection problem under Wasserstein uncertainty sets was investigated in \cite{BCZ22}, showing that this type of uncertainty leads to a regularization term in the mean-variance optimization problem. In \cite{PP18}, portfolio selection was studied under complete dependence uncertainty for a general risk functional satisfying subadditivity, comonotonic additivity, and positive homogeneity, giving rise to portfolio concentration. In our work, we study portfolio selection for non-convex risk functionals under a wide range of degrees of dependence uncertainty, with a particular focus on their consistency with the majorization order. } 

The paper makes four main contributions:
\begin{itemize}
    \item \textbf{A majorization-order-based consistency framework in the absence of model uncertainty.} We introduce a ``weak consistency'' notion for risk functionals with respect to the majorization order and show that weak consistency is equivalent to quasi-convexity in Theorem \ref{Th0}. This result clarifies the property needed for risk measures to support diversification (in a weak sense) with respect to the majorization order, complementing existing characterization results on convex-type risk measures.
    \item \textbf{Worst-case risk measure inequalities  and portfolio selection under complete dependence uncertainty.}  Under complete dependence uncertainty,  we establish inequalities with respect to the majorization order for the worst-case risk functionals satisfying subadditivity, and the worst-case  $\VaR$ and RVaR with marginal distributions possessing monotone densities in Theorem \ref{Th1}. It shows that the worst-case risk measure of the portfolios under dependence uncertainty is \emph{not} consistent with the majorization order for many commonly-used risk functionals.
     We further solve robust portfolio selection problems under dependence uncertainty for $\VaR$, RVaR, and standard deviation in Theorems \ref{prop1}-\ref{prop3}. Under broad conditions (including location-scale marginals or monotone-density classes), the optimal portfolio concentrates in a single asset, thereby generalizing the result for coherent risk measures. 
    
    \item \textbf{Robust portfolio selection under structured uncertainty.} We also develop the robust portfolio selection problem with uncertainty sets characterized by Wasserstein distance or moment information for distortion riskmetrics, and reduce the corresponding robust optimization problem to a tractable and deterministic optimization problem in Propositions \ref{th:wasserstein}-\ref{prop:moments}.
    \item \textbf{Robust aggregation with partial confidence.} {We propose a weighted robustness framework that interpolates between the reference and worst-case dependence structures. The resulting convex blend is structurally analogous to the constrained/unconstrained Expected Shortfall blend in the Fundamental Review of the Trading Book (\cite{BASEL19}), while remaining a distinct portfolio model.} 
\end{itemize}


To provide a clear overview of our findings, we summarize the relationship between uncertainty sets and the choice of risk measures, and the resulting optimal portfolio strategy in Table \ref{tab:summary}. {Under known dependence, the optimal strategy depends on expected losses, covariance, and the selected risk functional. Under complete dependence uncertainty, concentration becomes the robust choice for many widely-used risk functionals. Structured information, such as a Wasserstein ball \citep{PJ23,WLM26} or moment constraints \citep{EOO03,ZF09}, can introduce a penalty for concentration and thereby create a diversification incentive, but it does not guarantee an interior solution.}

\begin{table}[htbp]
\centering
\caption{Summary of Diversification vs. Concentration Effects}
\label{tab:summary}
\resizebox{\textwidth}{!}{%
\begin{tabular}{lll l}
\hline
\textbf{Uncertainty Regime} & \textbf{Risk Functionals} & \textbf{Optimal Portfolio} & \textbf{Key Driver} \\ \hline
Known Dependence & SD, $\ES$ (quasi-convexity) & {\textbf{Diversification/conditional}} & {
Exchangeability} \\ 
Full Ambiguity & $\ES$ (SA-PH-CA) & \textbf{Concentration} & Comonotonicity \\ 
Full Ambiguity & $\VaR, \text{RVaR}$ & \textbf{Concentration} & Monotone densities \\ 
Full Ambiguity & $\text{SD}$ & \textbf{Conditional} & Location-scale family \\ 
Wasserstein Ball & {Distortion riskmetrics} & {\textbf{Regularized/conditional}} & Norm penalty 
\\ 
Moment Constraints & {Distortion riskmetrics} & {\textbf{Regularized/conditional}} & Mean-variance trade-off \\ \hline
\end{tabular}}
\end{table}

{The practical message is twofold. First, diversification remains valuable, but its benefit relies on credible dependence modeling; when the copula is fully ambiguous, robust portfolio optimization often recommends concentration for many widely-used risk functionals unless additional information is incorporated. Second, interpolating between a reference dependence structure and the worst-case structure (and explicitly quantifying confidence in the joint model) provides a transparent way to balance diversification against robustness.}

The remainder of the paper is organized as follows. Section \ref{sec:notation} introduces the notation and basic definitions. In Section \ref{sec:known-dep}, we study risk measure inequalities under {a known dependence structure} and derive the necessary and sufficient condition for weak consistency. In Section \ref{sec:dep-uncertain}, we analyze worst-case risk measures under dependence uncertainty for $\VaR$, $\ES$, and RVaR, and establish the corresponding portfolio selection results, including conditions under which concentration is optimal in Section \ref{sec:portfolio}. We further extend the analysis to Wasserstein and moment-based uncertainty sets in Section \ref{sec:structured}. {Numerical illustrations of the main results are presented} in Section \ref{sec:numerical}. As a financial application, we introduce the weighted robustness formulation in Section \ref{sec:applications} and discuss its implications for diversification under uncertainty. 
Finally, Section \ref{sec:conclusion} concludes the paper.

\section{Notation and Preliminaries}\label{sec:notation}

Let $(\Omega,\mathcal F, \p)$ be an atomless probability space. {Let $\X$ be a linear subspace of $L^1(\Omega,\mathcal F,\p)$ that contains $L^\infty(\Omega,\mathcal F,\p)$. All expectations and risk-functional values below are assumed finite; results involving variances or covariances additionally impose the stated second-moment conditions.} For $X\in\X$, we use $F_X$ to denote its distribution function under $\p$. {We next introduce the risk measures used throughout the paper.} For a random variable $X\in\X$, we define its left quantile (also called Value-at-Risk ($\VaR$)) at level $\alpha \in (0,1]$ as
\begin{equation*}
    F_X^{-1}(\alpha) = \VaR_\alpha(X) = \inf\{x : F_X(x) \geq \alpha\},
\end{equation*}
and its right quantile at level $\alpha \in [0,1)$ as 
\begin{equation*}
    F_{X}^{-1,+}(\alpha) = \VaR^+_\alpha(X) = \inf\{x : F_X(x) > \alpha\}. 
\end{equation*}
For $\alpha \in [0,1)$, the Expected Shortfall ($\ES$) is defined as
\begin{equation*}
    \ES_{\alpha}(F_X) = \ES_{\alpha}(X) = \frac{1}{1-\alpha} \int_{\alpha}^{1} F_X^{-1}(t) \, \mathrm{d}t.
\end{equation*}
Note that $\VaR$ and $\ES$ are two popular regulatory risk measures widely applied in finance, insurance, economics and operations research. 
 The Range-Value-at-Risk (RVaR) proposed by \cite{CONT10} is a family of two-parameter risk measures bridging $\VaR$ and $\ES$, {defined, for $0 \leq \beta < \beta + \alpha \leq 1$, by}
\begin{equation*}
    R_{\beta,\alpha}(F_X) = R_{\beta,\alpha}(X) = \frac{1}{\alpha} \int_{\beta}^{\beta+\alpha} F_X^{-1,+}(1-t) \, \mathrm{d}t = \frac{1}{\alpha} \int_{1-\beta-\alpha}^{1-\beta} F_X^{-1}(t) \, \mathrm{d}t.
\end{equation*}
One can easily check that $\VaR$ and $\ES$ are the limiting cases of RVaR as follows.  For $\alpha \in (0,1)$, we have
\begin{equation*}
    \ES_{\alpha}(X) = R_{0,\alpha}(X), \quad \VaR_{\alpha}(X) = \lim_{\beta \downarrow 0} R_{1-\alpha, \beta}(X), \quad \text{and} \quad \VaR_{\alpha}^+(X) = \lim_{\beta \downarrow 0} R_{1-\alpha-\beta, \beta}(X).
\end{equation*}
Next, we introduce the concept of majorization to quantify the degree of diversification. For two vectors $\pmb{\lambda}, \pmb{\beta} \in \mathbb{R}^n$, we say $\pmb{\beta}$ is dominated by $\pmb{\lambda}$ in the \emph{majorization order}, denoted by $\pmb{\beta} \preceq \pmb{\lambda}$, if $\sum_{i=1}^{n} \phi(\beta_i) \leq \sum_{i=1}^{n} \phi(\lambda_i)$ for any continuous convex function $\phi$. A central characterization of this order is that $\pmb{\beta} \preceq \pmb{\lambda}$ if and only if there exists a \emph{doubly stochastic matrix} $\Lambda \in \mathcal{Q}_n$ such that $\pmb{\beta} = \Lambda \pmb{\lambda}$, where $\mathcal{Q}_n$ denotes the set of all $n \times n$ matrices with non-negative entries such that the sum of each row and each column equals one. 

Examples of doubly stochastic matrices include the {uniform averaging matrix} $\frac{1}{n} \mathbf{1}_n \mathbf{1}_n^\top$ (where $\mathbf{1}_n$ is the $n$-dimensional vector of ones) and any {$n\times n$ permutation matrix} $\Pi_k$. Specifically, for a fixed vector $\pmb{\lambda}$, $\Pi_k \pmb{\lambda}$ represents a permutation of its components, and $\Pi_1, \dots, \Pi_{n!}$ denote all possible permutation matrices. Conceptually, the operation $\Lambda \pmb{\lambda}$ represents a process of {``averaging,'' ``smoothing,'' or ``permuting''} the vector $\pmb{\lambda}$ (see Section 1.A.3 of \cite{MOA11}). We therefore use the majorization order as a mathematical proxy for the degree of diversification of a portfolio.

For two portfolio weight vectors $\pmb{\lambda}, \pmb{\beta} \in \mathbb{R}^n_+$, we say $\pmb{\beta}$ is more \emph{diversified} than $\pmb{\lambda}$ if $\pmb{\beta} \preceq \pmb{\lambda}$. For $\pmb{\lambda} \in \mathbb{R}^n_+$, a vector of random variables $\mathbf{X} \in \X^n$ and a risk functional {$\rho:\X\to\R$}, we define the portfolio risk functional as:
\begin{equation}
    S_\rho(\pmb{\lambda}; \mathbf{X}) = \rho\left( \pmb{\lambda}^\top \mathbf{X} \right).
\end{equation}
For simplicity, we write $S_\rho(\pmb{\lambda}) = S_\rho(\pmb{\lambda}; \mathbf{X})$. The primary objective of this paper is to investigate the conditions under which a more diversified portfolio induces a lower risk; that is, we seek to determine when $\pmb{\beta} \preceq \pmb{\lambda}$ implies $S_\rho(\pmb{\beta}) \leq S_\rho(\pmb{\lambda})$.  {Consistency with respect to the majorization order is commonly referred to as \textit{Schur-convexity} in the literature. However, the results developed below mainly concern a weaker form of monotonicity with respect to the majorization order. To distinguish our weaker property from Schur-convexity, we therefore do not use the term ``Schur-convexity'' in this paper.}



\section{Risk Inequalities with {Fixed} Dependence Structures}
\label{sec:known-dep}
In this section, we investigate whether a move toward a more diversified portfolio (quantified by the majorization order of the weights $\pmb{\lambda}$) necessarily results in a reduction of the portfolio risk $S_\rho(\pmb{\lambda})$.  We recall that a risk functional $\rho: \X \to \mathbb{R}$ is \emph{quasi-convex} if 
$\rho(\lambda X + (1-\lambda) Y) \leq \max\{\rho(X), \rho(Y)\}$ for all $X, Y \in \X$ and $\lambda \in [0,1]$.

Ideally, a risk manager would hope for \emph{strong consistency}:  a more diversified portfolio (in the sense of majorization order) leads to a lower, or at least not higher,  risk level, regardless of the underlying assets. However, as we show in Proposition \ref{prop0}, such a requirement is overly restrictive and leads to degenerate risk measures that are  not desirable in practical financial analysis.

\begin{proposition}[Strong Consistency]\label{prop0}
\begin{enumerate}
\item[(i)] $\rho \equiv c$ for some $c \in \mathbb{R}$ if and only if $\pmb{\beta} \preceq \pmb{\lambda}$ implies $S_\rho(\pmb{\beta}) \leq S_\rho(\pmb{\lambda})$ for all $\pmb{\lambda}, \pmb{\beta} \in \mathbb{R}_+^n$ and $X_1, \dots, X_n \in \X$;
\item[(ii)] $\rho(X) = f(\mathbb{E}[X])$ for some function $f: \mathbb{R} \to \mathbb{R}$ if and only if $\pmb{\beta} \preceq \pmb{\lambda}$ implies $S_\rho(\pmb{\beta}) \leq S_\rho(\pmb{\lambda})$ for all $\pmb{\lambda}, \pmb{\beta} \in \mathbb{R}_+^n$ and $X_1, \dots, X_n \in \X$ with the same mean.
\end{enumerate}
\end{proposition}
\begin{proof}  Note that the ``only if" parts for both (i) and (ii) are obvious. We next focus on the ``if" parts. Let $\pmb{\beta} = (1, 0, \dots, 0)$. Note that $\pmb{\beta} \preceq \Pi_k \pmb{\beta}$ and $\Pi_k \pmb{\beta} \preceq \pmb{\beta}$. Hence, $S_\rho(\pmb{\beta}) = S_\rho(\Pi_k \pmb{\beta})$ holds true for all $k=1, \dots, n!$, implying $\rho(X_1) = \dots = \rho(X_n)$ for all $X_1, \dots, X_n \in \X$ in (i), and for all $X_1, \dots, X_n \in \X$ with the same mean in (ii). Hence $\rho$ is a constant over $\X$ in (i) and a function of the mean in (ii).
\end{proof}

Proposition \ref{prop0} reveals a fundamental fact: {if we require a more diversified portfolio to have no greater risk for arbitrary underlying assets with the same mean}, then the adopted risk measure ignores the variance and other higher moments of the assets entirely, boiling down to a function of the mean. Such a requirement is too restrictive for real-world applications, where the primary goal of risk management is to capture the uncertainty and volatility of losses.

To obtain a more meaningful  class of risk measures, we relax our requirements and introduce the notion of weak consistency.  We say $\rho:\X\to\R$ is \emph{weakly consistent} if $\pmb{\beta} \preceq \pmb{\lambda}$ implies $S_\rho(\pmb{\beta}) \leq \max_{k=1, \dots, n!} S_\rho(\Pi_k \pmb{\lambda})$ for all $\pmb{\lambda}, \pmb{\beta} \in \mathbb{R}_+^n$ and $X_1, \dots, X_n \in \X$. Note that for $k=1,\dots,n!$, $\Pi_k \pmb{\lambda}$ and $\pmb{\lambda}$ are equivalent in the majorization order, i.e., $\Pi_k \pmb{\lambda}\sim \pmb{\lambda}$. 
Instead of requiring that a more diversified portfolio yields lower risk, as in strong consistency, weak consistency only requires that a more diversified portfolio {has no greater risk} than the portfolio corresponding to the \emph{worst-case permutation} of the original weights. {This relaxation reflects the fact that}, when assets are heterogeneous, the ordering of assets relative to weights matters.

\begin{theorem}[Weak Consistency]\label{Th0} A risk measure
 $\rho$ is  weakly consistent if and only if it is  quasi-convex.
\end{theorem}
\begin{proof} We first focus on the``if'' part. Note that $\pmb{\beta} \preceq \pmb{\lambda}$ is equivalent to $\pmb{\beta} = \Lambda \pmb{\lambda}$ for some $\Lambda \in \mathcal{Q}_n$. By Theorem 2.A.2 of \cite{MOA11}, for any $\Lambda \in \mathcal{Q}_n$, there exists $(w_1, \dots, w_{n!}) \in \Delta_{n!}$ such that $\Lambda = \sum_{k=1}^{n!} w_k \Pi_k$. Hence, 
\begin{align}\label{eq:1}
\pmb{\beta} = \sum_{k=1}^{n!} w_k \Pi_k \pmb{\lambda}.
\end{align}
We write $\Pi_k \pmb{\lambda} = (\lambda_{k_1}, \dots, \lambda_{k_n})$. Consequently,
\begin{align*}
S_{\rho}(\pmb{\beta}) = \rho\left(\sum_{i=1}^{n} \beta_i X_i\right) = \rho\left(\sum_{i=1}^{n} \sum_{k=1}^{n!} w_k \lambda_{k_i} X_i\right) = \rho\left(\sum_{k=1}^{n!} w_k \sum_{i=1}^{n} \lambda_{k_i} X_i\right) \\
\leq \max_{k=1, \dots, n!} \rho\left(\sum_{i=1}^{n} \lambda_{k_i} X_i\right) = \max_{k=1, \dots, n!} S_\rho(\Pi_k \pmb{\lambda}).
\end{align*}
Next we show the ``only if'' part.  Let $\pmb{\beta} = (\beta, 1-\beta, 0, \dots, 0)$ for some $\beta \in [0,1]$ and $\pmb{\lambda} = (1, 0, \dots, 0)$. Then by definition, we have $\pmb{\beta} \preceq \pmb{\lambda}$. It follows from weak consistency that $\rho(\beta X_1 + (1-\beta) X_2) = S_\rho(\pmb{\beta}) \leq \max_{k=1, \dots, n!} S_\rho(\Pi_k \pmb{\lambda}) = \max_{i=1, \dots, n} \rho(X_i)$. By setting $X_3 = X_4 = \dots = X_n = X_1$, we obtain $\rho(\beta X_1 + (1-\beta) X_2) \leq \max \{ \rho(X_1), \rho(X_2) \}$, which implies that $\rho$ is quasi-convex.
\end{proof}

Theorem \ref{Th0} provides a powerful insight: quasi-convexity (a much broader property than convexity) is the exact mathematical requirement  for the risk measures to support diversification in a weak sense. Economically, this means that as long as our risk measure does not exhibit erratic non-monotonicity under mixing, diversification will not lead to a risk level exceeding that of the portfolio with the worst-case assignment of the original weights.

To recover a stronger result where $\pmb{\beta} \preceq \pmb{\lambda}$ implies $S_\rho(\pmb{\beta}) \leq S_\rho(\pmb{\lambda})$ without falling into the degenerate cases described in Proposition \ref{prop0}, it is necessary to impose some specific structure on {the assets' negative returns}. The most natural assumption in this context is exchangeability. We say $\mathbf{X} \in (\mathcal{X})^n$ is {\emph{exchangeable}} if $\Pi_k\mathbf X\overset{d}{=}\mathbf{X}$ for all $k=1,\dots,n!$. {Exchangeability means that the asset negative-return vector is statistically identical under permutation of its components.} We say a risk functional $\rho:\X\to\R$ is \emph{law-invariant} if $\rho(X)=\rho(Y)$ whenever $X\overset{d}{=}Y$ for $X, Y\in\X$, where $\overset{d}{=}$ means equality in law. 

\begin{corollary}\label{cor2}
Suppose $\mathbf{X} \in \X^n$ is exchangeable and $\rho$ is quasi-convex and law-invariant. Then $\pmb{\beta} \preceq \pmb{\lambda}$ implies $S_{\rho}(\pmb{\beta}) \leq S_{\rho}(\pmb{\lambda})$ for all $\pmb{\lambda}, \pmb{\beta} \in \mathbb{R}_+^n$.
\end{corollary}

It is worth noting that a similar conclusion to Corollary \ref{cor2} was reached in Proposition 7.1 of \cite{CLLW22} under the assumption that $\rho$ is consistent with the convex order and $\mathbf{X}$ is exchangeable. {Corollary \ref{cor2} shows that a different property, quasi-convexity, is sufficient to guarantee consistency when the asset negative returns are exchangeable. More importantly, Theorem \ref{Th0} shows that quasi-convexity is necessary and sufficient for weak consistency for all $\mathbf{X}\in\X^n$.}

However, in most practical portfolios,  asset returns are inherently heterogeneous, and exchangeability is an overly restrictive assumption. Since strict consistency across all possible assets is only possible for the degenerate functionals identified in Proposition \ref{prop0}, we seek an alternative approach that preserves the benefits of the majorization order without requiring exchangability. We achieve this by modifying the risk functional to explicitly account for asset heterogeneity.  For $\rho:\X\to\R$ and $\pmb{\lambda}\in \mathbb{R}_+^n$, define the maximum permutation risk as
$$
S_{\rho}^{\max}(\pmb{\lambda}) = \max_{k=1, \dots, n!} S_\rho(\Pi_k \pmb{\lambda}).
$$
By defining the risk of a portfolio as the maximum risk of the portfolio  across all possible permutations of its weights, we construct a functional that is inherently consistent with the majorization order regardless of the properties of the underlying random vectors.

\begin{proposition}
Suppose $X_1, \dots, X_n \in \X$ and $\rho$ is quasi-convex.
Then $\pmb{\beta} \preceq \pmb{\lambda}$ implies $S_{\rho}^{\max}(\pmb{\beta}) \leq S_{\rho}^{\max}(\pmb{\lambda})$ for all $\pmb{\lambda}, \pmb{\beta} \in \mathbb{R}_+^n$.
\end{proposition}

\section{Risk Inequalities under Dependence Uncertainty}\label{sec:dep-uncertain}

 In the previous section, we established that, under a known dependence structure, quasi-convexity is {necessary and sufficient for weak consistency}. In practice, however, the estimation of the dependence structure is often inaccurate {because of limited data}, even when the marginal distributions can be estimated {with relatively high accuracy}. Moreover, in some cases, data from different correlated products are
separately collected so that no dependence information is available; see \cite{EPR13} and \cite{EWW15}. The misspecification of the dependence structure may result in severe consequences in risk management; see, e.g., \cite{MFE15}.
{Motivated by these considerations, we incorporate dependence uncertainty into the analysis of diversification and investigate how it influences diversification effects.}

We introduce some properties of risk functionals that will be used later. A risk functional $\rho: \X \to \mathbb{R}$ is \emph{subadditive (SA)} if $\rho(X+Y) \leq \rho(X) + \rho(Y)$ for all $X, Y \in \X$; it is \emph{positively homogeneous (PH)} if $\rho(\lambda X) = \lambda \rho(X)$ for all $X \in \X$ and $\lambda \geq 0$; and it is \emph{comonotonically additive (CA)} if $\rho(X+Y) = \rho(X) + \rho(Y)$ for all comonotonic $X, Y \in \X$\footnote{We say $X$ and $Y$ are comonotonic if there exist a random variable $Z$ and two {nondecreasing} functions $f$ and $g$ such that $X=f(Z)$ and $Y=g(Z)$ almost surely; see, e.g., Chapter 4 of \cite{FS16}.}. 

{Next, for $\mathbf F=(F_1,\dots,F_n)$, define the uncertainty set by}
$$ \mathcal{E}_n(\mathbf F)=\{F_{\mathbf X}: X_i\sim F_i,~i=1,\dots,n\}.$$
Under dependence uncertainty, we are concerned with the worst-case risk {over $\mathcal E_n(\mathbf F)$, defined by}
\begin{equation*}
    \overline{\rho}(\mathbf F) = \sup_{F_{\mathbf X}\in \mathcal{E}_n(\mathbf F)} \rho(X_1+\dots+X_n).
\end{equation*}
{Whenever $\overline\rho$ is expressed solely in terms of marginal distributions, we additionally assume that $\rho$ is law-invariant.} 
For $\lambda \geq 0$, let $F^\lambda$ denote the distribution of $\lambda X$ for $X \sim F$, and let $\mathbf{F}^{\pmb{\lambda}} = (F_1^{\lambda_1}, \dots, F_n^{\lambda_n})$ be the vector of scaled marginals corresponding to portfolio weights $\pmb{\lambda}$. Furthermore, let $\mathcal{M}_D^{\alpha}$ (respectively, $\mathcal{M}_I^{\alpha}$) denote the set of univariate distributions with decreasing (respectively, increasing) density beyond their $\alpha$-quantile.
\begin{remark}[On the Tail Behavior of Marginal Distributions]
The sets $\mathcal{M}_D^{\alpha}$ and $\mathcal{M}_I^{\alpha}$ are introduced to characterize the monotonicity of the probability density function (pdf) in the tail part of the distribution. Specifically, $\mathbf{F} \in (\mathcal{M}_D^{\alpha})^n$ implies that {the density of each $F_i$} is non-increasing for all $x \geq F_i^{-1}(\alpha)$. It is worth mentioning that the class $\mathcal{M}_D^{\alpha}$ encompasses the vast majority of risk factors encountered in finance and insurance. Most loss distributions are right-skewed with tails that decay toward zero; for instance, the normal, Student-$t$, {lognormal}, exponential, and Pareto distributions all exhibit decreasing densities beyond a certain quantile. 
Conversely, $\mathcal{M}_I^{\alpha}$ describes distributions where the probability mass clusters near the upper end of the support. Although such distributions are less common for individual asset losses, they may arise in certain insurance contracts with hard caps or liabilities with a high probability of total default. Mathematically, these conditions are crucial for the derivation of sharp convolution bounds in robust risk aggregation and {ensure} that the robust risk measures behave consistently under the majorization order; see \cite{CLLW22} and \cite{BLLW20} for more details. 
\end{remark}

The following theorem characterizes the behavior of worst-case risk measures under diversification.

\begin{theorem}\label{Th1} 
Suppose $\pmb{\lambda}, \pmb{\beta} \in \mathbb{R}_+^n$ and $\pmb{\beta} \preceq \pmb{\lambda}$.
\begin{enumerate}
\item[(i)] If $\rho$ is {law-invariant,} SA, CA, and PH, then $\min_{k=1,\dots, n!} \overline{\rho}(\mathbf{F}^{\Pi_k \pmb{\lambda}}) 
\leq \overline{\rho}(\mathbf{F}^{\pmb{\beta}})$;
\item[(ii)] If $\rho = \VaR_{\alpha}$ with $\alpha \in (0,1)$ and {$\mathbf{F} \in 
(\mathcal{M}_D^{\alpha})^n \cup (\mathcal{M}_I^{\alpha})^n$ or $n=2$}, then 
$\min_{k=1,\dots, n!} \overline{\VaR}_{\alpha}(\mathbf{F}^{\Pi_k \pmb{\lambda}}) 
\leq \overline{\VaR}_{\alpha}(\mathbf{F}^{\pmb{\beta}})$;
\item[(iii)] If $\rho = R_{\beta,\alpha}$ with $0 \leq \beta < \beta + \alpha \leq 1$ 
and $\mathbf{F} \in (\mathcal{M}_D^{1-\alpha-\beta})^n$, then 
$\min_{k=1,\dots, n!} \overline{R}_{\beta,\alpha}(\mathbf{F}^{\Pi_k \pmb{\lambda}}) 
\leq \overline{R}_{\beta,\alpha}(\mathbf{F}^{\pmb{\beta}})$.
\end{enumerate}
\end{theorem}

\begin{proof}
(i) Since $\rho$ is SA and PH, we have
$\overline{\rho}(\mathbf{F}^{\pmb{\beta}}) \leq \sum_{i=1}^n \rho(\beta_i X_i) =
\sum_{i=1}^n \beta_i \rho(X_i)$, where $X_i \sim F_i$. Note that  $(\beta_1 F_1^{-1}(U), \dots, \beta_n F_n^{-1}(U))$ is comonotonic, where $U\sim U[0,1]$.  Using the fact that $\rho$ is CA, we have $\rho(\beta_1
F_1^{-1}(U) + \dots + \beta_n F_n^{-1}(U)) = \sum_{i=1}^n \beta_i \rho(X_i)$. Consequently, we have $\overline{\rho}(\mathbf{F}^{\pmb{\beta}}) =
\sum_{i=1}^{n} \beta_i \rho(X_i)$.

Since $\pmb{\beta} \preceq \pmb{\lambda}$, by \eqref{eq:1} there exist weights
$(w_1, \dots, w_{n!}) \in \Delta_{n!}$ such that $\pmb{\beta} = \sum_{k=1}^{n!} w_k
\Pi_k \pmb{\lambda}$. This implies
\begin{align*}
\overline{\rho}(\mathbf{F}^{\pmb{\beta}}) = \sum_{i=1}^{n} \beta_i \rho(X_i)
= \sum_{k=1}^{n!} w_k \sum_{i=1}^{n} \lambda_{k_i} \rho(X_i)
= \sum_{k=1}^{n!} w_k \overline{\rho}(\mathbf{F}^{\Pi_k \pmb{\lambda}})
\geq \min_{k=1, \dots, n!} \overline{\rho}(\mathbf{F}^{\Pi_k \pmb{\lambda}}),
\end{align*}
where the last equality uses the fact that 
$\overline{\rho}(\mathbf{F}^{\Pi_k \pmb{\lambda}}) = \sum_{i=1}^n \lambda_{k_i}
\rho(X_i)$. This completes the proof of (i).

\noindent
(ii) 
For $\mathbf{F} \in (\mathcal{M}_D^{\alpha})^n \cup (\mathcal{M}_I^{\alpha})^n$ 
or $n=2$, in light of Theorem 2 of \cite{BLLW20}, we have
\begin{equation*}
    \overline{\VaR}_{\alpha}(\mathbf{F}) = \sup_{F_{\mathbf{X}} \in \mathcal{E}_n(\mathbf{F})} \VaR_{\alpha} \left( \sum_{i=1}^n X_i \right) = \inf_{\pmb{\gamma} \in (1-\alpha)\Theta_n} \sum_{i=1}^n R_{\gamma_i, \gamma_0}(X_i),
\end{equation*}
where $\Theta_n = \{ \pmb{\gamma} \in (0,1) \times [0,1)^n : \sum_{i=0}^{n} \gamma_i = 1 \}$ and $\pmb{\gamma} = (\gamma_0, \gamma_1, \dots, \gamma_n)$. 
Here, for all $\pmb{\beta} \in \mathbb{R}_+^n$ (including the case where some $\beta_i = 0$),
the convolution bound gives:
\begin{align*}
\overline{\VaR}_{\alpha}(\mathbf{F}^{\pmb{\beta}}) = \inf_{\pmb{\gamma} \in
(1-\alpha)\Theta_n} \sum_{i=1}^{n} \beta_i R_{\gamma_i, \gamma_0}(F_i).
\end{align*}
Substituting $\beta_i = \sum_{k=1}^{n!} w_k \lambda_{k_i}$ from \eqref{eq:1} in the above equation, we obtain
\begin{align*}
\overline{\VaR}_{\alpha}(\mathbf{F}^{\pmb{\beta}})
&= \inf_{\pmb{\gamma} \in (1-\alpha)\Theta_n} \sum_{k=1}^{n!} w_k
\sum_{i=1}^{n} \lambda_{k_i} R_{\gamma_i, \gamma_0}(F_i).
\end{align*}
Since the infimum of a convex combination is greater than or equal to the convex combination
of the corresponding individual infima, we have
\begin{align*}
\overline{\VaR}_{\alpha}(\mathbf{F}^{\pmb{\beta}})
&\geq \sum_{k=1}^{n!} w_k \inf_{\pmb{\gamma} \in (1-\alpha)\Theta_n}
\sum_{i=1}^{n} \lambda_{k_i} R_{\gamma_i, \gamma_0}(F_i) \\
&= \sum_{k=1}^{n!} w_k \overline{\VaR}_{\alpha}(\mathbf{F}^{\Pi_k \pmb{\lambda}}) 
\geq \min_{k=1, \dots, n!} \overline{\VaR}_{\alpha}(\mathbf{F}^{\Pi_k \pmb{\lambda}}).
\end{align*}

\noindent 
(iii) By Proposition 4 of \cite{FHLX23}, under $\mathbf{F} \in
(\mathcal{M}_D^{1-\alpha-\beta})^n$, the analogous convolution bound holds for
$\overline{R}_{\beta,\alpha}(\mathbf{F}^{\pmb{\beta}})$. The remainder of the argument
is identical to case (ii) with $\overline{\VaR}_\alpha$ replaced by
$\overline{R}_{\beta,\alpha}$. {The details of the proof are omitted.}
\end{proof}

The risk measures satisfying the conditions of Theorem \ref{Th1}(i) are of particular practical interest, as they include the Expected Shortfall, spectral risk measures, Gini-deviation, and mean-median deviation as important examples that are widely used in risk management, finance, and economics; {see, e.g.,} \cite{WWW20} and \cite{PWW25}.

Crucially, Theorem \ref{Th1} reveals a sharp contrast to the results in Section \ref{sec:known-dep}. It shows that the worst-case risk measure of the portfolios under dependence uncertainty is \emph{not} consistent with the majorization order for many commonly-used risk functionals. In fact, for $\pmb{\beta} \preceq \pmb{\lambda}$, there exists a permutation $\Pi_k$ such that the risk of the more concentrated portfolio corresponding to $\Pi_k\pmb{\lambda}$ is actually lower than or equal to the risk of the more diversified portfolio with weights $\pmb{\beta}$. This indicates that, under complete ambiguity of the copula, diversification may increase the robust risk for many widely-used risk functionals.

This effect becomes even more pronounced when the asset returns are identically distributed. If $\mathbf{F} = (F, \dots, F)$, then $\overline{\rho}(\mathbf{F}^{\Pi_k \pmb{\lambda}}) = \overline{\rho}(\mathbf{F}^{\pmb{\lambda}})$ for any permutation $\Pi_k$ and any $\rho$, leading to the following simplified results. 
\begin{proposition}
    \label{cor1} 
Suppose $\pmb{\lambda}, \pmb{\beta} \in \mathbb{R}_+^n$, $\mathbf{F} = (F, \dots, F)$ and $\pmb{\beta} \preceq \pmb{\lambda}$.
\begin{enumerate}
\item[(i)] If $\rho$ is law-invariant, SA, 
CA, 
and PH, 
then $\overline{\rho}(\mathbf F^{\pmb\lambda}) \leq \overline{\rho}(\mathbf F^{\pmb \beta})$;
\item[(ii)] If $\rho = \VaR_{\alpha}$ with $\alpha \in (0,1)$ and $F \in \mathcal{M}_D^{\alpha} \cup \mathcal{M}_I^{\alpha}$, then $\overline{\rho}(\mathbf F^{\pmb\lambda}) \leq \overline{\rho}(\mathbf F^{\pmb \beta})$;
\item[(iii)] If $\rho = R_{\beta, \alpha}$ with $0 \leq \beta < \beta + \alpha \leq 1$ and {$F \in \mathcal{M}_D^{1-\alpha-\beta}$}, then $\overline{\rho}(\mathbf F^{\pmb\lambda}) \leq \overline{\rho}(\mathbf F^{\pmb \beta})$.
\end{enumerate}
\end{proposition}
Proposition \ref{cor1} yields a counterintuitive conclusion: under dependence uncertainty and for homogeneous assets, the most diversified portfolio (e.g., the equally weighted portfolio) is the \emph{riskiest} with respect to worst-case risk measures, whereas the most concentrated portfolio is the \emph{safest}. Moreover, part (ii) of Proposition \ref{cor1} is consistent with Proposition 7.1 of \cite{CLLW22}.

To synthesize these findings, we define a risk functional that identifies the best possible assignment of weights to assets under the worst-case dependence {as follows}:
\begin{equation}
    \overline{S}_\rho^{\min}(\pmb{\lambda}) = \min_{k=1, \dots, n!} \overline{\rho}(\mathbf{F}^{\Pi_{k} \pmb{\lambda}}).
\end{equation}
As shown in the following result, the order induced by this functional is the reverse of the majorization order.

\begin{proposition}
Suppose $\pmb{\lambda}, \pmb{\beta} \in \mathbb{R}_+^n$ and $\pmb{\beta} \preceq \pmb{\lambda}$.
\begin{enumerate}
\item[(i)] If $\rho$ is 
law-invariant, SA, 
CA, 
and PH, 
then $\overline{S}_\rho^{\min}(\pmb{\lambda}) \leq \overline{S}_\rho^{\min}(\pmb{\beta})$ for all $\mathbf{F} \in \mathcal{M}^n$;
\item[(ii)] If $\rho = \VaR_{\alpha}$ with $\alpha \in (0,1)$, then $\overline{S}_\rho^{\min}(\pmb{\lambda}) \leq \overline{S}_\rho^{\min}(\pmb{\beta})$ for all $\mathbf{F} \in (\mathcal{M}_D^{\alpha})^n \cup (\mathcal{M}_I^{\alpha})^n$;
\item[(iii)] If $\rho = R_{\beta, \alpha}$ with $0 \leq \beta < \beta + \alpha \leq 1$, then $\overline{S}_\rho^{\min}(\pmb{\lambda}) \leq \overline{S}_\rho^{\min}(\pmb{\beta})$ for all {$\mathbf{F} \in (\mathcal{M}_D^{1-\alpha-\beta})^n$}.
\end{enumerate}
\end{proposition}

\section{Optimal Portfolio Selection under Dependence Uncertainty}\label{sec:portfolio}

In the preceding section, we established a theoretical connection between the majorization order of portfolio weights and the worst-case risk measures. Specifically, we showed that when the dependence structure is uncertain, diversification does not necessarily reduce worst-case risk measures and may, in many cases, increase them. This raises a fundamental question {for practitioners}: if the dependence structure among assets is uncertain during periods of market stress, what is the optimal way to allocate capital?

We focus on the risk-adjusted return maximization problem \eqref{RAR}. Throughout this section, we assume that {the marginal distributions $F_i$ of the underlying assets' negative returns}, for $i=1, \dots, n$, possess finite means. 
As a baseline, we consider risk functionals that possess strong structural properties. The following result, adapted from \cite{PP18}, demonstrates that for a wide class of risk measures, the optimal strategy is not to diversify, but to concentrate.

\begin{proposition}[Robust Concentration for SA-CA-PH Measures; adapted from \cite{PP18}]\label{prop:PflugPohl}
Suppose $\rho : \X \to \mathbb{R}$ is {law-invariant,} 
SA, 
CA, 
and PH, and $\kappa>0$. 
We have
\begin{equation}
    \sup_{\pmb{\lambda} \in \Delta_n} \inf_{F_{\mathbf{X}} \in \mathcal{E}_n(\mathbf{F})} \{ \mathbb{E}[-\pmb{\lambda}^\top \mathbf{X}] - \kappa\rho(\pmb{\lambda}^\top \mathbf{X}) \} = \max_{j=1, \dots, n} \{ \mathbb{E}[-X_j] - \kappa\rho(X_j) \},
\end{equation}
where $X_j \sim F_j$ with finite mean. In other words, the optimal portfolio concentrates all investment in a single asset $j^* \in \arg\max_{j} \{ \mathbb{E}[-X_j] - \kappa\rho(X_j) \}$.
\end{proposition}

Proposition \ref{prop:PflugPohl} is striking because it suggests that the ``optimal" robust portfolio is a vertex of the simplex $\Delta_n$ for risk functionals satisfying these three properties. {However, many widely-used risk functionals fail to satisfy all three properties simultaneously: $\VaR$ is generally not subadditive, whereas standard deviation (SD) is subadditive but not generally comonotonically additive.} A natural question is whether this concentration phenomenon persists for other commonly-used risk functionals that do not satisfy these three properties.


To test this, we examine the following risk functionals that are widely used in practice but {do not satisfy all three properties}: $\VaR^+$, $\mathrm{RVaR}$, and SD. We first focus on $\VaR^+_\alpha$, {the right quantile of the loss distribution}.

\begin{theorem}\label{prop1} 
For $\alpha \in (0,1)$ and $\kappa>0$, if $\mathbf{F} \in (\mathcal{M}_D^{\alpha})^n \cup (\mathcal{M}_I^{\alpha})^n$ or $n=2$, we have
\begin{align*}
\sup_{\boldsymbol{\lambda}\in\Delta_n}\inf_{F_{\mathbf X}\in \mathcal{E}_n(\mathbf F)}\left\{\mathbb{E} \left[-\boldsymbol\lambda^\top \mathbf X \right]-
\kappa \mathrm{VaR}_{\alpha}^+ \left(\boldsymbol\lambda^\top \mathbf X \right)\right\}=\mathbb{E} \left[- X_{j^*} \right]-
\kappa \mathrm{VaR}_{\alpha}^+ (X_{j^*}),
\end{align*}
for some $j^* \in \arg\max_{j=1,\dots, n}\{ \mathbb{E} \left[-X_{j} \right]- \kappa \mathrm{VaR}_{\alpha}^+ (X_{j})\}$, where $X_j \sim F_j$ with finite mean. 
\end{theorem}

\begin{proof}

Under the assumption that 
{$\mathbf{F} \in (\mathcal{M}_D^{\alpha})^n \cup (\mathcal{M}_I^{\alpha})^n$} 
or $n=2$, applying the convolution bound from Theorem 2 of \cite{BLLW20}, we have
\begin{equation}\label{eq:conv_bound}
    \sup_{F_{\mathbf{X}} \in \mathcal{E}_n(\mathbf{F})} \VaR_{\alpha}^+ \left( \sum_{i=1}^n \lambda_i X_i \right) = \inf_{\pmb{\gamma} \in (1-\alpha)\Theta_n} \sum_{i=1}^n \lambda_i R_{\gamma_i, \gamma_0}(X_i),
\end{equation}
where $\Theta_n = \{ \pmb{\gamma} \in (0,1) \times [0,1)^n : \sum_{i=0}^{n} \gamma_i = 1 \}$ and $\pmb{\gamma} = (\gamma_0, \gamma_1, \dots, \gamma_n)$.  
Then the optimization problem can be represented as 
\begin{equation}
\begin{aligned}
&\sup_{\boldsymbol{\lambda} \in \Delta_n} \inf_{F_{\mathbf{X}} \in \mathcal{E}_n(\mathbf{F})}
\left\{ \mathbb{E}[-\boldsymbol{\lambda}^\top \mathbf{X}] - \kappa
\mathrm{VaR}_\alpha^+(\boldsymbol{\lambda}^\top \mathbf{X}) \right\} \\
&= \sup_{\boldsymbol{\lambda} \in \Delta_n} \sup_{\pmb{\gamma} \in (1-\alpha)\Theta_n}
\sum_{i=1}^n \lambda_i \left( \mathbb{E}[-X_i] - \kappa R_{\gamma_i, \gamma_0}(X_i) \right)\\
&= \sup_{\pmb{\gamma} \in (1-\alpha)\Theta_n} \sup_{\boldsymbol{\lambda} \in \Delta_n}
\sum_{i=1}^n \lambda_i \left( \mathbb{E}[-X_i] - \kappa R_{\gamma_i,
\gamma_0}(X_i) \right), \label{eq:after_swap}
\end{aligned}
\end{equation}
where $X_i \sim F_i$. 
For any fixed $\pmb{\gamma}$, we denote by $\sum_{i=1}^n \lambda_i h_i(\pmb{\gamma})$ the objective in the right-hand side of \eqref{eq:after_swap}, 
where $h_i(\pmb{\gamma}) = \mathbb{E}[-X_i] - \kappa R_{\gamma_i, \gamma_0}(X_i)$, and this objective is linear in $\boldsymbol{\lambda}$. A linear
function over the compact convex polytope $\Delta_n$ attains its supremum at an extreme
point. Since the extreme points of $\Delta_n$ are precisely the unit vectors $\Delta_n^0 =
\{ \boldsymbol{\lambda} \in \Delta_n : \lambda_i \in \{0,1\},~i=1,\dots,n \}$, we have:
\begin{equation}
\sup_{\boldsymbol{\lambda} \in \Delta_n} \sum_{i=1}^n \lambda_i h_i(\pmb{\gamma})
= \max_{\boldsymbol{\lambda} \in \Delta_n^0} \sum_{i=1}^n \lambda_i h_i(\pmb{\gamma})
= \max_{j=1,\dots,n} h_j(\pmb{\gamma}). \label{eq:vertex}
\end{equation}
Substituting \eqref{eq:vertex} into \eqref{eq:after_swap}, and using the fact that
$\Delta_n^0$ is a finite set (the interchange between $\sup_{\pmb{\gamma}}$ and
$\max_{\boldsymbol{\lambda} \in \Delta_n^0}$ is always valid), we obtain
\begin{align*}
\eqref{eq:after_swap}
&= \max_{\boldsymbol{\lambda} \in \Delta_n^0} \sup_{\pmb{\gamma} \in (1-\alpha)\Theta_n}
\sum_{i=1}^n \lambda_i \left( \mathbb{E}[-X_i] - \kappa R_{\gamma_i,
\gamma_0}(X_i) \right) \\
&= \max_{\boldsymbol{\lambda} \in \Delta_n^0} \inf_{F_{\mathbf{X}} \in \mathcal{E}_n(\mathbf{F})}
\left\{ \mathbb{E}[-\boldsymbol{\lambda}^\top \mathbf{X}] - \kappa
\mathrm{VaR}_\alpha^+(\boldsymbol{\lambda}^\top \mathbf{X}) \right\} \\
&= \max_{j=1,\dots,n} \left\{ \mathbb{E}[-X_j] - \kappa \mathrm{VaR}_\alpha^+(X_j) \right\},
\end{align*}
where the last equality uses the fact that for $\boldsymbol{\lambda} \in \Delta_n^0$, the
portfolio reduces to a single asset, i.e., $\boldsymbol{\lambda}^\top \mathbf{X} = X_j$ for
some $j \in \{1,\dots,n\}$. This completes the proof.
\end{proof}

Theorem \ref{prop1} demonstrates that {the concentration phenomenon persists even for the non-coherent measure $\VaR^+$}, provided the marginal densities satisfy certain monotonicity conditions. This suggests that the concentration result is a fundamental consequence of the \emph{worst-case dependence structure} rather than the specific properties of the adopted risk measure. This finding is further reinforced when we consider $\mathrm{RVaR}$.

\begin{theorem}\label{prop2} 
For any $\alpha, \beta$ with $0 \leq \beta < \beta + \alpha \leq 1$, $\kappa>0$ and $\mathbf{F} \in (\mathcal{M}_D^{1-\alpha-\beta})^n$, we have
\begin{align*}
\sup_{\boldsymbol{\lambda}\in\Delta_n}\inf_{F_{\mathbf X}\in \mathcal{E}_n(\mathbf F)}\left\{\mathbb{E} \left[-\boldsymbol\lambda^\top \mathbf X \right]-
\kappa R_{\beta, \alpha}\left(\boldsymbol\lambda^\top \mathbf X \right)\right\}=\mathbb{E} \left[-X_{j^*} \right]-
\kappa R_{\beta, \alpha} (X_{j^*}),
\end{align*}
for some $j^* \in \arg\max_{j=1, \dots, n} \{ \mathbb{E} \left[-X_{j} \right]- \kappa R_{\beta, \alpha} (X_{j})\}$, where $X_j \sim F_j$ with finite mean. 
\end{theorem}

\begin{proof}
By Theorem 1 of \cite{BLLW20} and Proposition 4 of \cite{FHLX23}, we have
\begin{align*}
&\sup_{\boldsymbol{\lambda}\in\Delta_n}\inf_{F_{\mathbf X}\in \mathcal{E}_n(\mathbf F)}\left\{\mathbb{E} \left[-\boldsymbol\lambda^\top \mathbf X \right]-
\kappa R_{\beta, \alpha}\left(\boldsymbol\lambda^\top \mathbf X \right)\right\}\\
&=\sup_{\boldsymbol{\lambda}\in\Delta_n}
\sup_{\pmb{\gamma}\in(\beta+\alpha)\Theta_n,\gamma_0\geq \alpha}\left\{\mathbb{E} \left[-\boldsymbol\lambda^\top \mathbf X \right]-\sum_{i=1}^{n}\kappa R_{\gamma_i,\gamma_0}(\lambda_i X_i)\right\}.
\end{align*}
Further computation shows that
\begin{align*}
&\sup_{\boldsymbol{\lambda}\in\Delta_n}
\sup_{\pmb{\gamma}\in(\beta+\alpha)\Theta_n,\gamma_0\geq \alpha}\left\{\mathbb{E} \left[-\boldsymbol\lambda^\top \mathbf X \right]-\sum_{i=1}^{n}\kappa R_{\gamma_i,\gamma_0}(\lambda_i X_i)\right\}\\
&=
\sup_{\pmb{\gamma}\in(\beta+\alpha)\Theta_n,\gamma_0\geq \alpha}\sup_{\boldsymbol{\lambda}\in\Delta_n}\left\{\mathbb{E} \left[-\boldsymbol\lambda^\top \mathbf X \right]-\sum_{i=1}^{n}\kappa R_{\gamma_i,\gamma_0}(\lambda_i X_i)\right\}\\
&=
\sup_{\pmb{\gamma}\in(\beta+\alpha)\Theta_n,\gamma_0\geq \alpha}\max_{\boldsymbol{\lambda}\in\Delta_n^0}\left\{\sum_{i=1}^{n}\lambda_i(\mathbb{E}[-X_i]-\kappa R_{\gamma_i,\gamma_0}(X_i))\right\}\\
&=
\max_{\boldsymbol{\lambda}\in\Delta_n^0}\sup_{\pmb{\gamma}\in(\beta+\alpha)\Theta_n,\gamma_0\geq \alpha}\left\{\sum_{i=1}^{n}\lambda_i(\mathbb{E}[-X_i]-\kappa R_{\gamma_i,\gamma_0}(X_i))\right\}\\
&=
\max_{\boldsymbol{\lambda}\in\Delta_n^0}\inf_{F_{\mathbf X}\in \mathcal{E}_n(\mathbf F)}\left\{\mathbb{E} \left[-\boldsymbol\lambda^\top \mathbf X \right]-
\kappa R_{\beta, \alpha}\left(\boldsymbol\lambda^\top \mathbf X \right)\right\}\\
&=\max_{j=1,\dots, n} \{\mathbb{E} \left[-X_{j} \right]-
\kappa R_{\beta, \alpha} (X_{j})\}.
\end{align*}
\end{proof}

The consistent conclusions in Proposition \ref{prop:PflugPohl} and Theorems \ref{prop1}-\ref{prop2} highlight that the concentration phenomenon is remarkably robust across many commonly-used risk functionals in practice. {These results suggest} an important implication for robust portfolio management: under complete dependence uncertainty, diversification may no longer be desirable; instead, the optimal capital allocation may be to invest in a single asset. By concentrating, the investor eliminates the possibility of ``unfavorable'' dependence structures that may arise in diversified portfolios {during periods of market stress}.

Finally, we examine the case when the risk is evaluated by SD. Unlike the tail-risk measures discussed above, the SD is not comonotonic additive.

\begin{theorem}\label{prop3} 
For $\rho = \text{SD}$ and $\kappa>0$, if all $F_i$, $i=1, \dots, n$, have finite second moments, then we have
\begin{align*}
&\sup_{\boldsymbol{\lambda}\in\Delta_n}\inf_{F_{\mathbf X}\in \mathcal{E}_n(\mathbf F)}\left\{\mathbb{E} \left[-\boldsymbol\lambda^\top \mathbf X \right]-
\kappa \text{SD}\left(\boldsymbol\lambda^\top \mathbf X \right)\right\}\\
&=\sup_{\boldsymbol{\lambda}\in\Delta_n}\left\{\sum_{i=1}^{n}\lambda_i \mathbb{E}[-X_i]-
\kappa \text{SD}\left(\sum_{i=1}^n \lambda_i F_i^{-1}(U) \right)\right\},
\end{align*}
where $X_j \sim F_j$ with finite mean. Furthermore, if $F_i,\; i=1, \dots, n$, are in the same location-scale family, then we further have
\begin{align*}
&\sup_{\boldsymbol{\lambda}\in\Delta_n}\inf_{F_{\mathbf X}\in \mathcal{E}_n(\mathbf F)}\left\{\mathbb{E} \left[-\boldsymbol\lambda^\top \mathbf X \right]-
\kappa \text{SD}\left(\boldsymbol\lambda^\top \mathbf X \right)\right\}
=\mathbb{E}[-X_{j^*}]-
\kappa \text{SD}(X_{j^*}),
\end{align*}
for some $j^* \in \arg\max_{j=1, \dots, n} \{ \mathbb{E} \left[-X_{j} \right]- \kappa \text{SD}(X_{j})\}$.
\end{theorem}
\begin{proof}
Note that 
$$
\text{Var}\left(\boldsymbol\lambda^\top \mathbf X \right)=\sum_{i=1}^{n}\lambda_i^2\text{Var}(X_i)+\sum_{i\neq j} \lambda_i\lambda_j \text{Cov}(X_i,X_j).
$$
Using the fact that
$$
\text{Cov}(X_i,X_j)=\int_{\mathbb{R}}\int_{\mathbb{R}} \left[\mathbb{P}(X_i\leq x, X_j\leq y)-F_i(x)F_j(y)\right]\mathrm{d}x\mathrm{d}y,
$$
we have 
$$
\text{Cov}(X_i,X_j)\leq \text{Cov}(F_i^{-1}(U), F_j^{-1}(U)).
$$
Hence,
$$
\sup_{F_{\mathbf X}\in \mathcal{E}_n(\mathbf F)}\text{Var}\left(\boldsymbol\lambda^\top \mathbf X \right)=\text{Var}\left(\sum_{i=1}^n \lambda_i F_i^{-1}(U) \right),
$$
which implies the first claim.

If $F_i$, $i=1, \dots, n$, are in the same location-scale family, then
$$\text{Var}\left(\sum_{i=1}^n \lambda_i F_i^{-1}(U) \right)=\left(\sum_{i=1}^{n}\lambda_i \text{SD}(X_i)\right)^2.$$
Hence, we have
\begin{align*}
\sup_{\boldsymbol{\lambda}\in\Delta_n}\inf_{F_{\mathbf X}\in \mathcal{E}_n(\mathbf F)}\left\{\mathbb{E} \left[-\boldsymbol\lambda^\top \mathbf X \right]-
\kappa \text{SD}\left(\boldsymbol\lambda^\top \mathbf X \right)\right\}
&=\sup_{\boldsymbol{\lambda}\in\Delta_n}\left\{\sum_{i=1}^{n}\lambda_i\mathbb{E}[-X_i]-
\kappa \sum_{i=1}^{n}\lambda_i \text{SD}(X_i)\right\}\\
&=\sup_{\boldsymbol{\lambda}\in\Delta_n}\left\{\sum_{i=1}^{n}\lambda_i(\mathbb{E}[-X_i]-
\kappa \text{SD}(X_i))\right\}\\
&=\mathbb{E}[-X_{j^*}]-
\kappa \text{SD}(X_{j^*}),
\end{align*}
for some $j^* \in \arg\max_{j=1, \dots, n} \{ \mathbb{E} \left[-X_{j} \right]- \kappa \text{SD}(X_{j})\}$.
\end{proof}

Theorem \ref{prop3} reveals a nuanced boundary of the concentration phenomenon. While concentration remains optimal for assets with negative returns belonging to the same location-scale family, the optimal weights under $\text{SD}$ are more sensitive to the specific shapes of the marginal distributions. This suggests that, although the ``concentration paradox'' persists for many risk functionals, it may be mitigated when assets exhibit substantially different distributional characteristics, thereby potentially restoring the benefit of diversification.

Note that if $F_i$, $i=1,\dots,n$, do not belong to the same location-scale family, then the optimal decision variables in Theorem \ref{prop3} need not correspond to concentrated weights. The following example illustrates this phenomenon for $n=2$.
\begin{example}
{We consider a portfolio with only two assets. The negative return $X_1$ follows a normal distribution $\mathrm{N}(0, 1)$, and $X_2$ follows a Laplace (double-exponential) distribution with mean $0$ and variance $1$. The negative returns of both assets have the same mean and variance, but their distributions do not belong to the same location-scale family.}
{We now evaluate the robust objective}
\begin{equation*}
\inf_{F_{\mathbf X} \in \mathcal{E}_2(\mathbf{F})} \{ \mathbb{E}[-\boldsymbol{\lambda}^\top \mathbf{X}] - \kappa \mathrm{SD}(\boldsymbol{\lambda}^\top \mathbf{X}) \}= -\kappa \sup_{F_{\mathbf X} \in \mathcal{E}_2(\mathbf{F})} \mathrm{SD}(\boldsymbol{\lambda}^\top \mathbf{X}).
\end{equation*}
As established in the proof of Theorem \ref{prop3}, the worst-case SD is attained under comonotonicity, i.e.,
$$
\sup_{\mathbf{X} \in \mathcal{E}_2(\mathbf{F})} \text{SD}(\lambda_1 X_1 + \lambda_2 X_2) = \text{SD}(\lambda_1 F_1^{-1}(U) + \lambda_2 F_2^{-1}(U)).
$$
Let us evaluate this for a diversified portfolio $\boldsymbol{\lambda} = (0.5, 0.5)$ versus a concentrated portfolio $\boldsymbol{\lambda} = (1, 0)$.
\begin{enumerate}[(i)]
    \item For $\boldsymbol{\lambda} = (1, 0)$, the risk is $\text{SD}(F_1^{-1}(U)) = 1$.
    \item For $\boldsymbol{\lambda} = (0.5, 0.5)$, the risk is $\text{SD}(0.5 F_1^{-1}(U) + 0.5 F_2^{-1}(U))$. 
\end{enumerate}
{The quantile function of the normal distribution is $F_1^{-1}(u)=\Phi^{-1}(u)$, whereas the quantile function of a variance-one Laplace distribution is}
\begin{equation*}
F_2^{-1}(u)=-\frac{1}{\sqrt{2}}\operatorname{sgn}(u-0.5)\ln(1-2|u-0.5|).
\end{equation*}
{Let $c=\operatorname{Cov}(F_1^{-1}(U),F_2^{-1}(U))$. Both quantile functions are increasing and standardized, but they are not affine transforms of one another; hence $c<1$. Therefore,}
\begin{equation*}
\operatorname{Var}\!\left(0.5F_1^{-1}(U)+0.5F_2^{-1}(U)\right)=\frac{1+c}{2}<1.
\end{equation*}
{Consequently,}
$$
-\kappa \cdot \text{SD}(0.5 F_1^{-1}(U) + 0.5 F_2^{-1}(U)) > -\kappa \cdot 1.
$$
{Because both vertices have risk one while the midpoint has strictly smaller risk, no vertex is optimal; the optimal weight vector $\boldsymbol{\lambda}^*$ is diversified.} This demonstrates that the concentration phenomenon for SD depends critically on the location-scale homogeneity of the negative returns of the assets.
\end{example}

We can also consider the case where a risk manager seeks to maximize the objective subject to a fixed {target expected return} $\mu\in\R$. We define $\Delta_n(\mu) = \Delta_n \cap \{ \boldsymbol{\lambda} : \mathbb{E} [-\boldsymbol\lambda^\top \mathbf X ] = \mu \}$. Then the portfolio selection problem becomes
$$\sup_{\boldsymbol{\lambda}\in\Delta_n(\mu)}\inf_{F_{\mathbf X}\in \mathcal{E}_n(\mathbf F)}\left\{\mathbb{E} \left[-\boldsymbol\lambda^\top \mathbf X \right]-
\kappa \rho\left(\boldsymbol\lambda^\top \mathbf X \right)\right\},$$
which is equivalent to
\begin{equation*}
    \inf_{\boldsymbol{\lambda} \in \Delta_n(\mu)} \sup_{F_{\mathbf X} \in \mathcal{E}_n(\mathbf F)} \rho (\boldsymbol\lambda^\top \mathbf X).
\end{equation*}
In this case, portfolio concentration may be mitigated by imposing a suitable expected return constraint. We can easily simplify the above optimization problem by applying the same arguments used for  $\rho=\VaR$, $\mathrm{RVaR}$ or  SD in Theorems \ref{prop1}-\ref{prop3}. The details are omitted. 

\section{Optimal Portfolio Selection under Structured Uncertainty}\label{sec:structured}

In the previous sections, we demonstrated that, under complete dependence uncertainty with fixed marginal distributions, the robust optimal portfolio strategy for many commonly-used risk functionals is to concentrate the portfolio in a single asset. This result arises because, in the absence of information about the joint distribution, the worst-case dependence structure penalizes diversification.

In practice, however, a risk manager may face other forms of uncertainty. For example, they may possess a reference model, information on the first two moments, or specific constraints on the expected return. We now explore how these more structured forms of uncertainty affect the optimal portfolio selection.

\subsection{Wasserstein Uncertainty Sets}
One popular notion used in mass transportation and distributionally robust optimization is the Wasserstein metric. 
For two $n$-dimensional distributions $F$ and $G$, $a, p \geq 1$, the Wasserstein metric is defined as
\begin{equation*}
    d_{a,p}^{n}(F, G) = \inf_{\mathbf X \sim F, \mathbf Y \sim G} \left( \mathbb{E} [ \|\mathbf X - \mathbf Y\|_a^p ] \right)^{1/p},
\end{equation*}
where $\|\cdot\|_a$ represents the {$\ell_a$ norm}.
If the reference distribution is denoted by $F_0$ satisfying suitable moment conditions, for $\epsilon > 0$,
we define the Wasserstein ball around $F_0$ as $\mathcal{M}_{a,p,\epsilon}^{n}(F_0) = \{ F : d_{a,p}^{n}(F, F_0) \leq \epsilon \}$.
For more details about {the definition and properties of the Wasserstein metric}, we refer to \cite{EK18} and \cite{BM19}. Moreover, we refer to \cite{BCZ22} for the mean--variance portfolio optimization under Wasserstein uncertainty sets.

Let $\mathcal H$ denote the set of all functions $g:\R\to\R$  with bounded variation satisfying $g(0)=g(0+)=0$ and $g(1)=g(1-)$. For $g\in\mathcal H$, the distortion riskmetric $\rho_g$ is defined as
\begin{equation*}
    \rho_g(X) = \int_{0}^{\infty} g(\mathbb{P}(X > x)) \mathrm{d}x + \int_{-\infty}^{0} [g(\mathbb{P}(X > x)) - g(1)] \mathrm{d}x.
\end{equation*}
The distortion riskmetric introduced in \cite{WWW20a} and \cite{WWW20}  is a general class of risk functionals including the classical distortion risk measures, Gini-deviation and mean-median deviation as important examples.

We next present the results for the portfolio optimization problem under distortion riskmetrics with Wasserstein uncertainty sets. For $g\in\mathcal H$, let $g_{\kappa}(x) = x + \kappa g(x)$, $x\in [0,1]$, for some {$\kappa>0$}.
\begin{proposition}\label{th:wasserstein}
Suppose {$\kappa > 0$, $p>1$, and $g\in\mathcal H$ with concave $g_\kappa$}. Then we have
\begin{align*}
\sup_{\boldsymbol\lambda \in \Delta_n} \inf_{F_{\mathbf X} \in \mathcal{M}_{a,p,\epsilon}^n(F_0)} \left\{ \mathbb{E} \left[ -\boldsymbol\lambda^\top \mathbf X \right] - \kappa \rho_g(\boldsymbol\lambda^\top \mathbf X) \right\} = \sup_{\boldsymbol\lambda \in \Delta_n} \left\{ -\rho_{g_\kappa}(\boldsymbol{\lambda}^\top \mathbf X_0) - \epsilon \left( \int_0^1 |1 + \kappa g'(t)|^q \mathrm{d}t \right)^{1/q} \|\boldsymbol\lambda\|_b \right\},
\end{align*}
where  $\mathbf X_0\sim F_0$, $1/a + 1/b = 1$ and $1/p + 1/q = 1$.
\end{proposition}

\begin{proof}
Following Theorem 5 of \cite{MWW24}, we have $\mathcal{M}_{1,p,\epsilon\|\boldsymbol\lambda\|_b}^{1}(F_{\boldsymbol{\lambda}^\top \mathbf X_0}) = \{ F_{\boldsymbol\lambda^\top \mathbf X} : F_{\mathbf X} \in \mathcal{M}_{a,p,\epsilon}^n(F_0) \}$. Thus, the optimization problem boils down to
\begin{align*}
\sup_{\boldsymbol\lambda \in \Delta_n} \left\{ -\sup_{F_Y \in \mathcal{M}_{1,p,\epsilon\|\boldsymbol\lambda\|_b}^{1}(F_{\boldsymbol{\lambda}^\top \mathbf X_0})} \rho_{g_{\kappa}}(Y) \right\}.
\end{align*}
Since $g_{\kappa}$ is concave, by Proposition 4 of \cite{LMWW22}, we have
\begin{equation*}
    \sup_{F_Y \in \mathcal{M}_{1,p,\epsilon\|\boldsymbol\lambda\|_b}^{1}(F_{\boldsymbol{\lambda}^\top \mathbf X_0})} \rho_{g_\kappa}(Y) = \rho_{g_{\kappa}}(\boldsymbol{\lambda}^\top \mathbf X_0) + \epsilon \left( \int_0^1 |1 + \kappa g'(t)|^q \mathrm{d}t \right)^{1/q} \|\boldsymbol\lambda\|_b.
\end{equation*}
This completes the proof.
\end{proof}

{Unlike the previous sections, the term $\|\boldsymbol{\lambda}\|_b$ can penalize concentration, depending on $b$. Thus, when a reference distribution is available, the Wasserstein ambiguity radius acts as a regularization parameter and can discourage (but does not categorically rule out) a concentrated optimum. This effect becomes stronger as the ambiguity radius $\epsilon$ increases.}

\subsection{Moment-Based Ambiguity}
A widely-used uncertainty set in finance, risk management and operations research is defined by the moment information of the underlying distribution when only the first two moments are reliable. We define this uncertainty set  as 
$$
D(\pmb{\mu}, \Sigma) = \{ F_{\mathbf X}: \mathbb{E}[X_i] = \mu_i, \text{Cov}(\mathbf X) = \Sigma \},
$$
where $\Sigma$ is {an $n\times n$ positive-definite matrix}. {For scalar $m\in\mathbb R$ and $s\geq0$, define the univariate moment class}
\begin{equation*}
D_1(m,s)=\{F_Y:\mathbb E[Y]=m,\ \operatorname{SD}(Y)=s\}.
\end{equation*}
For a distortion function $g \in \mathcal{H}$, we let $g^*$ denote its concave envelope defined as $g^*=\inf\{h\in\mathcal H: h \text{ is concave over } [0,1] \text{ and }  h\geq g \}$.

\begin{proposition}\label{prop:moments}
{Suppose $v_{g^*}<\infty$. The robust portfolio selection problem under moment uncertainty reduces to}
\begin{align*}
\sup_{\boldsymbol{\lambda} \in \Delta_n} \inf_{F_{\mathbf X} \in D(\pmb{\mu}, \Sigma)} \left\{ \mathbb{E} [-\boldsymbol\lambda^\top \mathbf X ] - \kappa \rho_g (\boldsymbol\lambda^\top \mathbf X) \right\} = \sup_{\boldsymbol{\lambda} \in \Delta_n} \left\{ -(1 + \kappa g(1))\boldsymbol\lambda^\top\boldsymbol\mu - \kappa v_{g^*}\sqrt{\boldsymbol\lambda^\top{\Sigma}\boldsymbol\lambda} \right\},
\end{align*}
with $v_{g^*}=\sqrt{\int_0^1((g^*)'(t)-g(1))^2\d t}$.
\end{proposition}

\begin{proof}
In light of \cite{P07}, we have
\begin{equation*}
\left\{F_{\boldsymbol\lambda^\top\mathbf X}: \mathbb E(X_i)=\mu_i, \operatorname{Cov}(\mathbf X)=\Sigma\right\}=D_1\left(\boldsymbol\lambda^\top\boldsymbol\mu,\sqrt{\boldsymbol\lambda^\top\Sigma\boldsymbol\lambda}\right).
\end{equation*}
 Hence, the inner minimization problem can be rewritten as
\begin{equation*}
    \inf_{F_Y \in {D_1\left(\boldsymbol\lambda^\top\boldsymbol\mu,\sqrt{\boldsymbol\lambda^\top\Sigma\boldsymbol\lambda}\right)}} \{ -\mathbb{E}[Y] - \kappa \rho_g(Y) \}.
\end{equation*}
In light of Theorem 5 of \cite{PWW25}, we have 
\begin{align*}
    \inf_{F_Y \in D_1\left(\boldsymbol\lambda^\top\boldsymbol\mu,\sqrt{\boldsymbol\lambda^\top\Sigma\boldsymbol\lambda}\right)} \{ -\mathbb{E}[Y] - \kappa \rho_g(Y) \} &= -\boldsymbol\lambda^\top\boldsymbol\mu - \kappa \sup_{F_Y \in D_1\left(\boldsymbol\lambda^\top\boldsymbol\mu,\sqrt{\boldsymbol\lambda^\top\Sigma\boldsymbol\lambda}\right)} \rho_g(Y) \\
    &=-\boldsymbol\lambda^\top\boldsymbol\mu - \kappa \left(g(1)\boldsymbol\lambda^\top\boldsymbol\mu+v_{g^*}\sqrt{\boldsymbol\lambda^\top{\Sigma}\boldsymbol\lambda}\right)  \\
    &= -(1 + \kappa g(1))\boldsymbol\lambda^\top\boldsymbol\mu - \kappa v_{g^*}\sqrt{\boldsymbol\lambda^\top{\Sigma}\boldsymbol\lambda}.
\end{align*}
Hence, the optimization problem reduces to
\begin{equation*}
    \sup_{\boldsymbol{\lambda} \in \Delta_n} \left\{  -(1 + \kappa g(1))\boldsymbol\lambda^\top\boldsymbol\mu - \kappa v_{g^*}\sqrt{\boldsymbol\lambda^\top{\Sigma}\boldsymbol\lambda}\right\},
\end{equation*}
which completes the proof.
\end{proof}
The result in Proposition \ref{prop:moments} reveals a mean-variance-type structure. {The objective reflects a trade-off between expected loss, $\boldsymbol\lambda^\top\boldsymbol\mu$, and portfolio volatility, $\sqrt{\boldsymbol{\lambda}^\top \Sigma \boldsymbol{\lambda}}$, thereby introducing a Markowitz-style diversification incentive without guaranteeing an interior optimum.}

\section{Numerical Illustrations}\label{sec:numerical}


{This section provides a sensitivity analysis of the robust objective
\begin{equation}\label{eq:numerical-objective}
V_\rho(\kappa)=\sup_{\boldsymbol\lambda\in\Delta_3}
\inf_{F_{\mathbf X}\in\mathcal E_3(\mathbf F)}
\left\{\mathbb E[-\boldsymbol\lambda^\top\mathbf X]
-\kappa\rho(\boldsymbol\lambda^\top\mathbf X)\right\},
\end{equation}
where every $X_i$ is measured as a loss (larger values are worse).  The parameters are deliberately chosen so that no asset simultaneously has the smallest mean loss and the smallest risk.  Let $Z_1,Z_2,Z_3\sim\mathrm N(0,1)$, with their joint distribution left unspecified, and define
\begin{equation}\label{eq:numerical-marginals}
X_1=4+2.6Z_1,\qquad X_2=4.8+0.9Z_2,\qquad
X_3=5.147425+0.092541\exp(Z_3).
\end{equation}
Thus $X_3$ is a shifted and scaled lognormal loss.  The two constants in $X_3$ are selected so that $\mathbb E[X_3]=5.3$ and $\mathrm{SD}(X_3)=0.2$.  The normal and lognormal densities are decreasing beyond the quantile levels used below, so the marginal assumptions of Theorems \ref{prop1}--\ref{prop2} are satisfied.  Table \ref{table:numerical-marginals} displays the resulting summary statistics.

\begin{table}[htbp]
\centering
\caption{Marginal inputs and risk values in the numerical illustration.}
\label{table:numerical-marginals}
\def\arraystretch{1.25}
\begin{tabular}{c|c|c|c|c}
\hline
Asset & $\mathbb E[X_i]$ & $\mathrm{SD}(X_i)$ & $\VaR_{0.99}^+(X_i)$ & $R_{0.01,0.10}(X_i)$ \\
\hline
1 & $4.0000$ & $2.6000$ & $10.0485$ & $8.1960$ \\
2 & $4.8000$ & $0.9000$ & $6.8937$ & $6.2525$ \\
3 & $5.3000$ & $0.2000$ & $6.0951$ & $5.6327$ \\
\hline
\end{tabular}
\end{table}

\subsection{Tail-risk measures: switches among concentrated portfolios}

For $\rho\in\{\VaR_{0.99}^+,R_{0.01,0.10}\}$, Theorems \ref{prop1}-\ref{prop2} reduce \eqref{eq:numerical-objective} to the upper envelope of three affine functions,
\begin{equation}\label{eq:numerical-envelope}
V_\rho(\kappa)=\max_{i=1,2,3}\{-m_i-\kappa r_i^\rho\},
\qquad m_i=\mathbb E[X_i],\quad r_i^\rho=\rho(X_i).
\end{equation}
If Assets $i$ and $j$ are adjacent on this envelope, their switching point is
$\kappa_{ij}=(m_j-m_i)/(r_i^\rho-r_j^\rho)$.  The exact calculations produce the regimes in Table \ref{table:numerical-regimes}.  At each displayed threshold the two adjacent vertices tie; away from the thresholds the maximizing vertex is unique.  This presentation separates the concentration result from the identity of the selected asset: increasing risk aversion changes the selected asset, but not the vertex structure of the reported concentrated optimizer.

\begin{center}
\begin{minipage}{0.95\textwidth}
\centering
\captionof{table}{Risk-aversion regimes for the optimal concentrated portfolio.  Here $\mathbf e_i$ denotes full investment in Asset $i$.}
\label{table:numerical-regimes}
\def\arraystretch{1.25}
\begin{tabular}{c|c|c|c}
\hline
Risk functional & $\boldsymbol\lambda^*=\mathbf e_1$ & $\boldsymbol\lambda^*=\mathbf e_2$ & $\boldsymbol\lambda^*=\mathbf e_3$ \\
\hline
$\VaR_{0.99}^+$ & $0\leq\kappa<0.2536$ & $0.2536<\kappa<0.6261$ & $\kappa>0.6261$ \\
$R_{0.01,0.10}$ & $0\leq\kappa<0.4116$ & $0.4116<\kappa<0.8067$ & $\kappa>0.8067$ \\
\hline
\end{tabular}
\end{minipage}
\end{center}

Panels (a) and (b) of Figure \ref{fig:numerical-sensitivity} plot the three asset-specific values in \eqref{eq:numerical-envelope}.  The thick black upper envelope is the robust value.  It makes both switching points visible and avoids basing the conclusion on a small, arbitrary list of values of $\kappa$.

\subsection{Worst-case SD: a heterogeneous-shape counterpoint}

For SD, Theorem \ref{prop3} shows that the worst case is the comonotonic coupling.  Under that coupling we may take $Z_1=Z_2=Z_3=Z$ in \eqref{eq:numerical-marginals}; the two normal losses are affine in $Z$, while the third loss is affine in $\exp(Z)$.  Since
$\mathrm{Corr}(Z,\exp(Z))=1/\sqrt{e-1}=0.7629$, the comonotonic covariance matrix is
\begin{equation}\label{eq:comonotonic-cov-numerical}
\Sigma_c=
\begin{pmatrix}
6.7600&2.3400&0.3967\\
2.3400&0.8100&0.1373\\
0.3967&0.1373&0.0400
\end{pmatrix}.
\end{equation}
Consequently, maximizing the robust mean--SD objective is equivalent to minimizing
$\boldsymbol\lambda^\top\mathbf m+\kappa
\sqrt{\boldsymbol\lambda^\top\Sigma_c\boldsymbol\lambda}$ over $\Delta_3$.  This is a convex program and was solved over a fine grid of $\kappa$ values.  The optimizer is $\mathbf e_1$ for $0\leq\kappa<0.4706$, is $\mathbf e_2$ for $0.4706<\kappa\leq0.6690$, is a mixture of Assets 2 and 3 for $0.6690<\kappa<1.0276$, and is $\mathbf e_3$ for $\kappa\geq1.0276$.  At $\kappa=0.4706$, every mixture of Assets 1 and 2 is optimal because their comonotonic normal losses are perfectly correlated.  Representative solutions are reported in Table \ref{table:numerical-sd}; see panel (c) of Figure \ref{fig:numerical-sensitivity} for a full illustration.

\begin{table}[htbp]
\centering
\caption{Representative optimal portfolios for the worst-case mean--SD objective.}
\label{table:numerical-sd}
\def\arraystretch{1.20}
\begin{tabular}{c|c|c}
\hline
$\kappa$ & $\boldsymbol\lambda^*$ & $V_{\mathrm{SD}}(\kappa)$ \\
\hline
$0$    & $(1,0,0)$                 & $-4.0000$ \\
$0.25$ & $(1,0,0)$                 & $-4.6500$ \\
$0.50$ & $(0,1,0)$                 & $-5.2500$ \\
$0.75$ & $(0,0.2035,0.7965)$        & $-5.4395$ \\
$1.00$ & $(0,0.0082,0.9918)$        & $-5.4999$ \\
$1.25$ & $(0,0,1)$                 & $-5.5500$ \\
\hline
\end{tabular}
\end{table}

\begin{figure}[htbp]
\centering
\includegraphics[width=\textwidth]{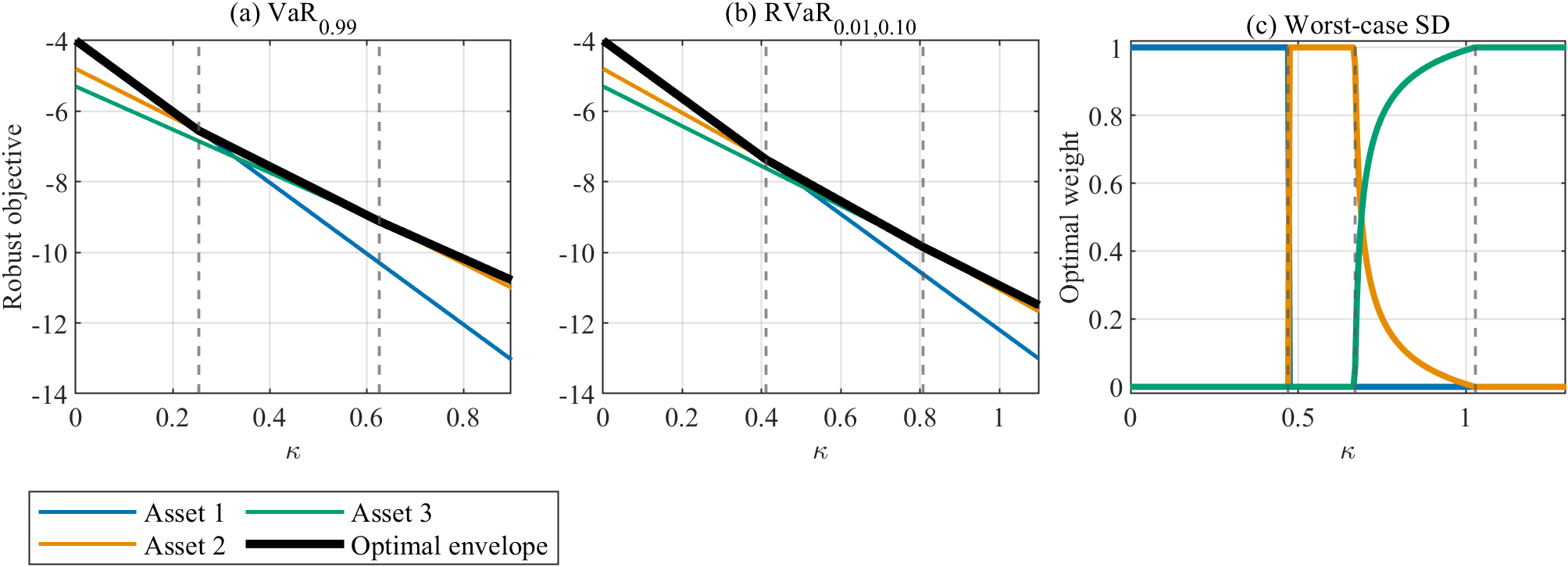}
\caption{Sensitivity to risk aversion.  Panels (a) and (b) show the asset-specific objectives and their optimal envelope under $\VaR_{0.99}^+$ and $R_{0.01,0.10}$; the dashed lines are the analytical switching thresholds.  Panel (c) shows the optimal weights under worst-case SD.  The interval in which Assets 2 and 3 are both held illustrates why Theorem \ref{prop3} requires a common location--scale family for its concentration conclusion.}
\label{fig:numerical-sensitivity}
\end{figure}

\subsection{The value of dependence information}

To isolate the effect of dependence uncertainty, we retain exactly the same three marginals but compare two SD models.  The reference model assumes independence and hence has covariance matrix
$\Sigma_0=\operatorname{diag}(6.76,0.81,0.04)$, whereas the fully robust model uses the comonotonic covariance matrix $\Sigma_c$ in \eqref{eq:comonotonic-cov-numerical}.  For $d\in\{0,c\}$, define the minimized mean-SD criterion
\begin{equation}\label{eq:sd-dependence-comparison}
J_d^*(\kappa)=\min_{\boldsymbol\lambda\in\Delta_3}
\left\{\boldsymbol\lambda^\top\mathbf m+
\kappa\sqrt{\boldsymbol\lambda^\top\Sigma_d\boldsymbol\lambda}\right\}.
\end{equation}
Because $\Sigma_c$ represents the worst-case covariances for nonnegative weights, the robustness premium
$\Delta(\kappa)=J_c^*(\kappa)-J_0^*(\kappa)$ is nonnegative.  We also report the Herfindahl-Hirschman concentration index
$\mathrm{HHI}(\boldsymbol\lambda)=\sum_i\lambda_i^2$, which ranges from $1/3$ for equal weights to $1$ for a fully concentrated portfolio.

\begin{table}[htbp]
\centering
\caption{Dependence sensitivity of the optimal mean--SD portfolio.  The subscripts $0$ and $c$ denote independence and worst-case comonotonic dependence, respectively.}
\label{table:dependence-comparison}
\def\arraystretch{1.20}
\resizebox{0.95\textwidth}{!}{%
\begin{tabular}{c|c|c|c|c|c}
\hline
$\kappa$ & $\boldsymbol\lambda_0^*$ & $\boldsymbol\lambda_c^*$ & $\mathrm{HHI}_0$ & $\mathrm{HHI}_c$ & $\Delta(\kappa)$ \\
\hline
$0.50$ & $(0.3280,0.6720,0)$       & $(0,1,0)$              & $0.5592$ & $1.0000$ & $0.1897$ \\
$0.75$ & $(0.2214,0.7126,0.0660)$  & $(0,0.2035,0.7965)$    & $0.5612$ & $0.6758$ & $0.1372$ \\
$1.00$ & $(0.0591,0.2120,0.7289)$  & $(0,0.0082,0.9918)$    & $0.5797$ & $0.9838$ & $0.0977$ \\
$1.25$ & $(0.0416,0.1580,0.8003)$  & $(0,0,1)$              & $0.6672$ & $1.0000$ & $0.0832$ \\
\hline
\end{tabular}}
\end{table}

At $\kappa=0.50$, the independent model allocates to both Assets 1 and 2, while complete dependence ambiguity changes the decision to full investment in Asset 2.  At $\kappa=0.75$, both solutions are diversified, but they diversify for different reasons and across different assets: the independent solution uses all three assets, whereas the robust solution mixes only Assets 2 and 3 because their comonotonic correlation is below one.  For $\kappa\geq1.0276$, the robust allocation is fully concentrated in Asset 3, while the independent solution remains diversified.  Panels (a) and (b) of Figure \ref{fig:numerical-extensions} show the full weight paths and the corresponding HHI values.  The premium $\Delta(\kappa)$ measures the optimized increase in the mean--SD criterion caused by guarding against unknown dependence; it is not a transaction cost or a regulatory capital charge.

\subsection{Sensitivity to the shape of the third marginal}

The preceding calibration uses a lognormal shape parameter equal to one.  To check whether the tail-risk switching pattern is an artifact of this choice, we vary the shape parameter $\eta$ over $[0.2,1.5]$ while holding the mean and SD of Asset 3 fixed.  Specifically, let
\begin{equation}\label{eq:shape-sensitivity-marginal}
X_3(\eta)=a(\eta)+b(\eta)\exp(\eta Z_3),\qquad
b(\eta)=\frac{0.2}{\sqrt{(e^{\eta^2}-1)e^{\eta^2}}},\qquad
a(\eta)=5.3-b(\eta)e^{\eta^2/2}.
\end{equation}
Then $\mathbb E[X_3(\eta)]=5.3$ and $\mathrm{SD}(X_3(\eta))=0.2$ for every $\eta$ in the experiment.  Hence, any change in the VaR or RVaR boundary is caused by distributional shape rather than by the first two moments.  The density remains decreasing beyond the relevant $0.89$ and $0.99$ quantiles throughout the reported range.

\begin{table}[htbp]
\centering
\caption{Tail-shape sensitivity for Asset 3.  The quantities $\kappa_{23}$ are the risk-aversion levels at which the maximizing vertex switches from Asset 2 to Asset 3.}
\label{table:tail-shape-sensitivity}
\def\arraystretch{1.20}
\begin{tabular}{c|c|c|c|c}
\hline
$\eta$ & $\VaR_{0.99}^+(X_3)$ & $\kappa_{23}^{\VaR}$ & $R_{0.01,0.10}(X_3)$ & $\kappa_{23}^{\mathrm{RVaR}}$ \\
\hline
$0.20$ & $5.8553$ & $0.4815$ & $5.6523$ & $0.8331$ \\
$0.40$ & $5.9438$ & $0.5264$ & $5.6708$ & $0.8596$ \\
$0.55$ & $6.0034$ & $0.5616$ & $5.6753$ & $0.8663$ \\
$0.70$ & $6.0517$ & $0.5938$ & $5.6705$ & $0.8592$ \\
$1.00$ & $6.0951$ & $0.6261$ & $5.6327$ & $0.8067$ \\
$1.30$ & $6.0458$ & $0.5897$ & $5.5635$ & $0.7257$ \\
$1.50$ & $5.9617$ & $0.5365$ & $5.5083$ & $0.6719$ \\
\hline
\end{tabular}
\end{table}

The Asset 1-to-Asset 2 thresholds remain $0.2536$ for VaR and $0.4116$ for RVaR because those two marginals are unchanged.  The Asset 2-to-Asset 3 thresholds vary over $[0.4815,0.6261]$ for VaR and $[0.6719,0.8663]$ for RVaR.  The dependence on $\eta$ is non-monotone: once mean and SD are fixed, increasing lognormal skewness does not necessarily increase a fixed $0.99$ quantile or a fixed quantile average.  Nevertheless, the three-region structure (Asset 1 at low $\kappa$, Asset 2 at intermediate $\kappa$, and Asset 3 at high $\kappa$) persists throughout the experiment, as shown in panels (c) and (d) of Figure \ref{fig:numerical-extensions}.

\begin{figure}[htbp]
\centering
\includegraphics[width=\textwidth]{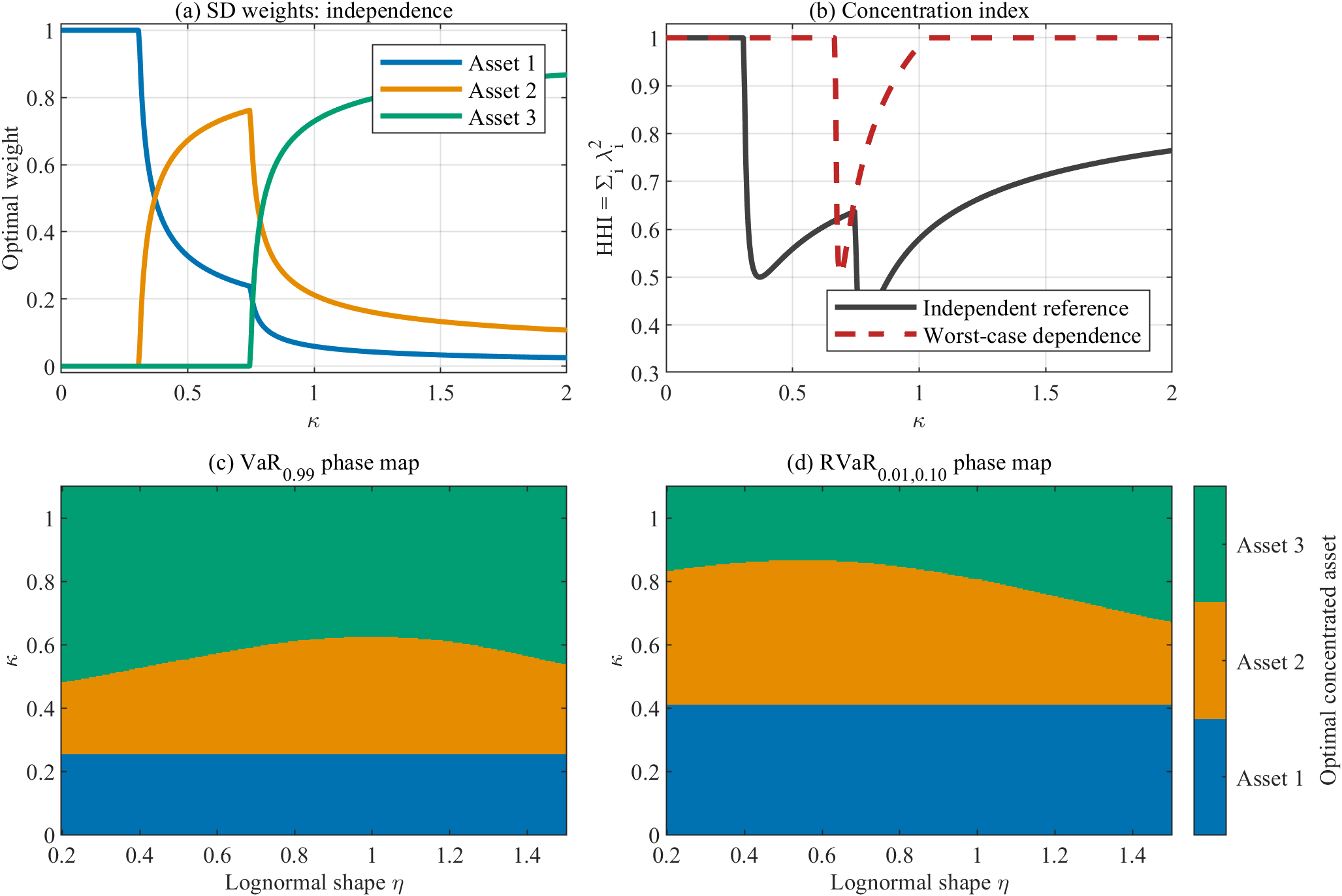}
\caption{Additional sensitivity experiments.  Panel (a) gives the optimal SD weights under the independent reference model.  Panel (b) compares the HHI of the independent and worst-case SD allocations.  Panels (c) and (d) give the identity of the optimal concentrated asset as a function of risk aversion $\kappa$ and the lognormal shape $\eta$, while the mean and SD of Asset 3 remain fixed.}
\label{fig:numerical-extensions}
\end{figure}

Taken together, the experiments make three distinct points.  First, complete dependence ambiguity forces a concentrated optimizer for the two tail-risk measures, even though the selected asset changes with $\kappa$.  Second, heterogeneous marginal shapes can make a diversified portfolio strictly preferable for worst-case SD over an intermediate range of risk aversion.  Third, comparison with the independent reference model shows directly that removing credible dependence information can sharply increase concentration, while the tail-shape experiment demonstrates that the VaR and RVaR regime pattern is robust to substantial changes in skewness. 
}

\section{Application: A Weighted Robustness Framework}\label{sec:applications}
Throughout the previous sections, we have analyzed two extreme forms of dependence modeling: Problem \eqref{eq:opt1}, where the joint distribution $F$ is completely specified, and Problem \eqref{eq:opt2}, where only the marginal distributions are known and the dependence structure is completely unspecified, leading to a worst-case analysis. Neither of these two cases is fully realistic in  practical applications. The former ignores dependence uncertainty, while the latter often leads to overly conservative (concentrated) positions that may not align with a firm's actual risk appetite or historical data.

To bridge this gap, we propose a weighted robustness model. {For a robustness weight $\omega \in [0,1]$}, we consider the following optimization problem:
\begin{equation}\label{eq:opt4}
	\min_{\boldsymbol{\lambda} \in \Delta_n} \left( {\omega} \sup_{F_{\mathbf Y} \in \mathcal{E}_n(\mathbf F)} \rho(\boldsymbol{\lambda}^\top \mathbf Y) + {(1-\omega)} \rho(\boldsymbol{\lambda}^\top \mathbf X) \right).
\end{equation}
{In this formulation, $\omega$ represents the weight assigned to robustness against dependence misspecification. When $\omega=0$, the manager relies entirely on the specified model $\mathbf{X}$; when $\omega=1$, the manager adopts a fully robust and worst-case posture. For $0 < \omega < 1$, the objective balances the reference model against the worst-case model.}

{A motivation for this approach is its structural resemblance to the constrained/unconstrained Expected Shortfall blend used in the Fundamental Review of the Trading Book (FRTB); see \cite{BASEL19}. The resemblance is conceptual rather than an identity of formulas.}
The worst-case $\ES$ under dependence uncertainty is analytically tractable and the objective in \eqref{eq:opt4} simplifies to:
\begin{equation}\label{eq:opt5}
	\min_{\boldsymbol{\lambda} \in \Delta_n} \left( {\omega} \sum_{i=1}^n \lambda_i \ES_p(X_i) + {(1-\omega)} \ES_p(\boldsymbol{\lambda}^\top \mathbf X) \right).
\end{equation}
This optimization problem becomes computationally tractable when $\mathbf{X}$ follows a multivariate normal distribution $\mathrm{N}(\boldsymbol{\mu}, \Sigma)$ or, more generally, a multivariate elliptical distribution. The same structural simplification holds for any coherent and CA risk measure, allowing the risk manager to analytically trade off the diversification benefits of the joint model against the additive penalty of the worst-case scenario. {Thus, our formulation provides a tractable portfolio model for studying the trade-off between reference-model and worst-case risk. Implementing an actual FRTB capital calculation would additionally require the prescribed regulatory aggregation rules.}

{We illustrate the result with a concrete example.} Let $\mathbf{X} \sim \text{N}(\boldsymbol{\mu}, \Sigma)$ with $n = 3$ assets, where:
\begin{equation*}
    \boldsymbol{\mu} = \begin{pmatrix} 4.9 \\ 4.1 \\ 5.5 \end{pmatrix}, \quad
    \Sigma = \begin{pmatrix}
        1.8^2 & 0.3 \times 1.8 \times 0.2 & 0.2 \times 1.8 \times 1.0 \\
        0.3 \times 1.8 \times 0.2 & 0.2^2 & 0.25 \times 0.2 \times 1.0 \\
        0.2 \times 1.8 \times 1.0 & 0.25 \times 0.2 \times 1.0 & 1.0^2
    \end{pmatrix},
\end{equation*}
corresponding to $X_1 \sim \text{N}(4.9, 1.8^2)$, $X_2 \sim \text{N}(4.1, 0.2^2)$,
and $X_3 \sim \text{N}(5.5, 1.0^2)$, with pairwise correlations $\rho_{12} = 0.30$,
$\rho_{13} = 0.20$, and $\rho_{23} = 0.25$. 
For $X \sim \text{N}(\mu, \sigma^2)$, the $\ES$ at confidence level
$p \in (0,1)$ admits the closed-form expression:
\begin{equation}\label{eq:ES_normal}
    \ES_p(X) = \mu + \sigma \cdot c_p, \quad \text{where} \quad
    c_p = \frac{\phi(\Phi^{-1}(p))}{1-p},
\end{equation}
and $\phi$ and $\Phi$ denote the standard normal PDF and CDF, respectively. Since $\mathbf{X}$ is multivariate normal, the portfolio $\boldsymbol{\lambda}^\top
\mathbf{X} \sim \text{N}(\boldsymbol{\lambda}^\top \boldsymbol{\mu},\,
\boldsymbol{\lambda}^\top \Sigma \boldsymbol{\lambda})$ for any $\boldsymbol{\lambda}
\in \Delta_n$. Applying \eqref{eq:ES_normal}, the marginal and portfolio ES are:
\begin{align}
    \ES_p(X_i) &= \mu_i + \sigma_i \cdot c_p, \quad i = 1, \dots, n, \label{eq:ES_marginal}\\
    \ES_p(\boldsymbol{\lambda}^\top \mathbf{X}) &= \boldsymbol{\lambda}^\top \boldsymbol{\mu}
    + \sqrt{\boldsymbol{\lambda}^\top \Sigma \boldsymbol{\lambda}} \cdot c_p. \label{eq:ES_portfolio}
\end{align}
Substituting \eqref{eq:ES_marginal} and \eqref{eq:ES_portfolio} into \eqref{eq:opt5},
and using $\sum_{i=1}^n \lambda_i \mu_i = \boldsymbol{\lambda}^\top \boldsymbol{\mu}$,
the objective simplifies to
\begin{equation}\label{eq:obj_simplified}
    \min_{\boldsymbol{\lambda} \in \Delta_n} \left\{
    \boldsymbol{\lambda}^\top \boldsymbol{\mu} + c_p \left(
    {\omega} \, \boldsymbol{\lambda}^\top \boldsymbol{\sigma} +
    {(1-\omega)} \sqrt{\boldsymbol{\lambda}^\top \Sigma \boldsymbol{\lambda}}
    \right) \right\},
\end{equation}
where $\boldsymbol{\sigma} = (\sigma_1, \dots, \sigma_n)^\top$ is the vector of
marginal standard deviations.


{The objective \eqref{eq:obj_simplified} has three components: (i) $\boldsymbol{\lambda}^\top \boldsymbol{\mu}$, the expected loss of the portfolio; (ii) $\omega c_p\boldsymbol{\lambda}^\top \boldsymbol{\sigma}$, the linear worst-case risk contribution; and (iii) $(1-\omega)c_p\sqrt{\boldsymbol{\lambda}^\top \Sigma \boldsymbol{\lambda}}$, the joint-model risk contribution. The second component alone favors an asset with the smallest marginal standard deviation. The third component creates a diversification incentive when $\Sigma$ is positive definite, but the full mean--risk objective can still have a boundary minimizer.}

Here are some special cases:
\begin{enumerate}
    \item {Full trust in the joint model} ({$\omega = 0$}): The objective reduces to:
    \begin{equation*}
        \min_{\boldsymbol{\lambda} \in \Delta_n}
        \left\{ \boldsymbol{\lambda}^\top \boldsymbol{\mu} +
        c_p \sqrt{\boldsymbol{\lambda}^\top \Sigma \boldsymbol{\lambda}} \right\},
    \end{equation*}
    {which is the classical mean--ES minimization problem. Its solution is often diversified, but it need not be interior for arbitrary $\boldsymbol\mu$ and $\Sigma$.}

    \item {Full robustness} ({$\omega = 1$}): The objective reduces to:
    \begin{equation*}
        \min_{\boldsymbol{\lambda} \in \Delta_n}
        \left\{ \boldsymbol{\lambda}^\top \boldsymbol{\mu} +
        c_p \, \boldsymbol{\lambda}^\top \boldsymbol{\sigma} \right\}
        = \min_{\boldsymbol{\lambda} \in \Delta_n}
        \sum_{i=1}^n \lambda_i (\mu_i + c_p \sigma_i)
        = \min_{i=1,\dots,n} \ES_p(X_i),
    \end{equation*}
    which is minimized by concentrating entirely in the asset $j^* = \arg\min_i
    \ES_p(X_i)$. For our parameters, $j^* = 2$ since $\ES_p(X_2) = 4.1 + 0.2 \cdot c_p$
    is the smallest.

    \item {Intermediate} ({$0 < \omega < 1$}): {The optimizer need not move smoothly between regimes because its active set may change. It is characterized by the full Karush-Kuhn-Tucker conditions. In particular, for multipliers $\nu\in\mathbb R$ and $\boldsymbol\eta\geq \mathbf{0}$,}
    \begin{equation*}
\boldsymbol\mu+c_p\left(\omega\boldsymbol\sigma+(1-\omega)
        \frac{\Sigma\boldsymbol\lambda^*}{\sqrt{(\boldsymbol\lambda^*)^\top
        \Sigma\boldsymbol\lambda^*}}\right)-\nu\mathbf 1_n-\boldsymbol\eta=\mathbf0,
        \qquad \eta_i\lambda_i^*=0.
    \end{equation*}
    {Together with primal and dual feasibility, these conditions determine numerically any value of $\omega$ at which the active set changes.}
\end{enumerate}

{We further consider the case $\rho=\mathrm{SD}$. For fixed marginals and nonnegative portfolio weights, the worst-case variance is attained by comonotonic random variables. Hence, the weighted model balances the SD under the reference dependence structure against the SD under comonotonic dependence. For elliptical reference models, the resulting objective resembles \eqref{eq:obj_simplified}. Increasing $\omega$ strengthens the linear worst-case term and may move the optimizer toward concentration, but neither a U-shaped profile nor a monotone or smooth transition is guaranteed without additional assumptions.}
{The cases $\rho=\VaR_\beta$ and $\rho=\mathrm{RVaR}_{\beta_1,\beta_2}$ require separate analysis because VaR and RVaR lack some of the convexity properties used above.}

\section{Concluding Remarks}\label{sec:conclusion}

In this paper, we have investigated the fundamental tension between portfolio diversification and concentration under dependence uncertainty. By introducing a majorization-order-based framework, we show that, while diversification remains an effective tool for risk reduction in the absence of model uncertainty, its benefit may diminish or even reverse when dependence uncertainty is taken into account.

{We have proven that quasi-convexity is a necessary and sufficient condition for a risk functional to be weakly consistent with the majorization order. Through the analysis of a broad range of risk functionals, including $\VaR$, $\ES$, RVaR, and SD, we demonstrated a ``concentration paradox'': under full dependence uncertainty, a robust optimal portfolio may concentrate the investment in a single asset. Structured information, such as a Wasserstein ambiguity set or moment constraints, can introduce a diversification incentive, although it does not guarantee an interior optimum. Finally, we proposed a weighted robustness framework that is structurally analogous to, but not identical with, the FRTB Expected Shortfall blend and that balances a reference model against a worst-case model.}

Our approach offers a distinct perspective on uncertainty in risk management compared with {the related work of} \cite{FLW24} and \cite{LWY26}. While these two articles treat uncertainty as a \emph{subjective judgment}, where an analyst weights different risk models based on confidence, we treat it as \emph{objective} (or model-based) ambiguity. Operating within a robust optimization paradigm, we view the uncertainty set $\mathcal{E}_n(\mathbf{F})$ as {representing a lack of information about the copula}, aiming to protect the portfolio against the worst-case dependence structure. By focusing on the \emph{structural properties} of portfolio weights rather than the weighting of models, we derive concrete implications for allocation, {including conditions favoring concentration rather than diversification}.

{Ultimately, this paper highlights that, although diversification remains a cornerstone of modern portfolio theory, its effectiveness must be evaluated in light of model uncertainty. Under complete dependence uncertainty and the assumptions of our concentration results, an optimal robust strategy can be to invest in the best single asset rather than to spread risk across assets. Future research could extend this framework by incorporating partial dependence information, such as constraints on correlations or other dependence characteristics.}

\vskip -0.2cm
\quad 

\noindent
{\bf Acknowledgements.}
YL acknowledges financial support from the National Natural Science Foundation of China (Grant No. 12401624), Guangdong Science and Technology Program (Grant No. 2024QN11X076), Shenzhen Science and Technology Program (Grant Nos. RCBS20231211090814028, JCYJ20250604141203005, 2025TC0010) and The Chinese University of Hong Kong (Shenzhen) University Development Fund (Grant No. UDF01003336) and is partly supported by the Guangdong Provincial Key Laboratory of Mathematical Foundations for Artificial Intelligence (Grant No. 2023B1212010001). 

\end{document}